\RequirePackage[l2tabu,orthodox]{nag}
\documentclass
[11pt,letterpaper]
{article}

\usepackage[utf8]{inputenc}
\usepackage{algorithmic}
\usepackage{algorithm}

\usepackage{etex}
\usepackage{xspace,enumerate}
\usepackage[dvipsnames]{xcolor}
\usepackage[T1]{fontenc}
\usepackage[full]{textcomp}
\usepackage[american]{babel}
\usepackage{mathtools}
\usepackage{thmtools}
\usepackage{thm-restate}
\usepackage{amsthm}
\newtheorem{theorem}{Theorem}[section]
\newtheorem*{theorem*}{Theorem}

\newtheorem{proposition}[theorem]{Proposition}
\newtheorem*{proposition*}{Proposition}
\newtheorem{lemma}[theorem]{Lemma}
\newtheorem*{lemma*}{Lemma}
\newtheorem{corollary}[theorem]{Corollary}
\newtheorem*{conjecture*}{Conjecture}
\newtheorem{fact}[theorem]{Fact}
\newtheorem*{fact*}{Fact}

\newtheorem*{hypothesis*}{Hypothesis}

\theoremstyle{definition}
\newtheorem{definition}[theorem]{Definition}
\newtheorem*{definition*}{Definition}
\newtheorem{algorithms}{Algorithm}

\newtheorem{model}[theorem]{Model}

\theoremstyle{remark}

\newtheorem*{claim*}{Claim}
\newtheorem{remark}[theorem]{Remark}
\newtheorem*{remark*}{Remark}

\newtheorem*{observation*}{Observation}
\usepackage[
letterpaper,
top=1.2in,
bottom=1.2in,
left=1in,
right=1in]{geometry}
\usepackage{newpxtext} 
\usepackage{textcomp} 
\usepackage[varg,bigdelims]{newpxmath}
\usepackage[scr=rsfso]{mathalfa}
\usepackage{bm} 
\let\mathbb\varmathbb
\usepackage{microtype}
\usepackage[
pagebackref,
colorlinks=true,
urlcolor=blue,
linkcolor=blue,
citecolor=OliveGreen,
]{hyperref}
\usepackage[capitalise,nameinlink]{cleveref}
\crefname{lemma}{Lemma}{Lemmas}
\crefname{fact}{Fact}{Facts}
\crefname{theorem}{Theorem}{Theorems}
\crefname{corollary}{Corollary}{Corollaries}
\crefname{claim}{Claim}{Claims}
\crefname{example}{Example}{Examples}
\crefname{algorithms}{Algorithm}{Algorithms}
\crefname{problem}{Problem}{Problems}
\crefname{definition}{Definition}{Definitions}
\usepackage{paralist}
\usepackage{turnstile}
\usepackage{mdframed}
\usepackage{tikz}
\usepackage{caption}
\DeclareCaptionType{Algorithm}
\usepackage{newfloat}

\newcommand{\Authornotecolored}[3]{}

\newcommand{\SK}[1]{}
\newcommand{\Pnote}[1]{}

\usepackage{boxedminipage}
\newcommand{\paren}[1]{(#1)}
\newcommand{\Paren}[1]{\left(#1\right)}

\newcommand{\set}[1]{\{#1\}}
\newcommand{\Set}[1]{\left\{#1\right\}}

\newcommand{\norm}[1]{\lVert#1\rVert}

\newcommand{\iprod}[1]{\langle#1\rangle}

\newcommand{\Esymb}{\mathbb{E}}
\newcommand{\Psymb}{\mathbb{P}}

\DeclareMathOperator*{\E}{\Esymb}

\DeclareMathOperator*{\ProbOp}{\Psymb}
\renewcommand{\Pr}{\ProbOp}

\newcommand{\from}{\colon}

\newcommand{\mper}{\,.}
\newcommand{\mcom}{\,,}
\newcommand\bdot\bullet

\DeclareMathOperator{\sign}{sign}

\newcommand{\N}{\mathbb N}
\newcommand{\R}{\mathbb R}

\newcommand{\cA}{\mathcal A}
\newcommand{\cB}{\mathcal B}
\newcommand{\cC}{\mathcal C}

\newcommand{\cI}{\mathcal I}

\newcommand{\cN}{\mathcal N}
\newcommand{\cO}{\mathcal O}

\newcommand{\cR}{\mathcal R}
\newcommand{\cS}{\mathcal S}

\newcommand{\calD}{\mathcal{D}}
\renewcommand{\leq}{\leqslant}

\renewcommand{\geq}{\geqslant}
\renewcommand{\ge}{\geqslant}
\let\epsilon=\varepsilon
\numberwithin{equation}{section}
\newcommand\MYcurrentlabel{xxx}
\newcommand{\MYstore}[2]{%
  \global\expandafter \def \csname MYMEMORY #1 \endcsname{#2}%
}
\newcommand{\MYload}[1]{%
  \csname MYMEMORY #1 \endcsname%
}
\newcommand{\MYnewlabel}[1]{%
  \renewcommand\MYcurrentlabel{#1}%
  \MYoldlabel{#1}%
}
\newcommand{\MYdummylabel}[1]{}
\newcommand{\torestate}[1]{%
  \let\MYoldlabel\label%
  \let\label\MYnewlabel%
  #1%
  \MYstore{\MYcurrentlabel}{#1}%
  \let\label\MYoldlabel%
}
\newcommand{\restatetheorem}[1]{%
  \let\MYoldlabel\label
  \let\label\MYdummylabel
  \begin{theorem*}[Restatement of \cref{#1}]
    \MYload{#1}
  \end{theorem*}
  \let\label\MYoldlabel
}
\newcommand{\restatelemma}[1]{%
  \let\MYoldlabel\label
  \let\label\MYdummylabel
  \begin{lemma*}[Restatement of \cref{#1}]
    \MYload{#1}
  \end{lemma*}
  \let\label\MYoldlabel
}
\newcommand{\restateprop}[1]{%
  \let\MYoldlabel\label
  \let\label\MYdummylabel
  \begin{proposition*}[Restatement of \cref{#1}]
    \MYload{#1}
  \end{proposition*}
  \let\label\MYoldlabel
}
\newcommand{\restatefact}[1]{%
  \let\MYoldlabel\label
  \let\label\MYdummylabel
  \begin{fact*}[Restatement of \prettyref{#1}]
    \MYload{#1}
  \end{fact*}
  \let\label\MYoldlabel
}
\newcommand{\restate}[1]{%
  \let\MYoldlabel\label
  \let\label\MYdummylabel
  \MYload{#1}
  \let\label\MYoldlabel
}

\newcommand{\eps}{\epsilon}

\allowdisplaybreaks
\newcommand*{\zo}{\set{0,1}}

\newcommand*{\on}{\{\pm 1\}}

\DeclareMathOperator{\pE}{\tilde{\mathbb{E}}}

\newcommand{\1}{\bm{1}}

\newcommand{\tmu}{\tilde{\mu}}
\newcommand{\wt}{\mathsf{wt}}

\renewcommand{\S}{\mathbb{S}}
\newcommand{\Sol}{\mathrm{Sol}}

\newcommand{\Lin}{\mathrm{Lin}}
\newcommand{\blfootnote}[1]{}

\title{
 List-Decodable Linear Regression
}

\author{
Sushrut Karmalkar\thanks{University of Texas at Austin. Supported by NSF Award CNS-1414023}
\and
Adam R. Klivans \thanks{University of Texas at Austin. Supported by NSF Award CCF-1717896}
\and
  Pravesh K. Kothari\thanks{Princeton University and Institute for Advanced Study. Supported by Schmidt Foundation Fellowship and Avi Wigderson's NSF Award CCF-1412958.}}

\begin{document}

\pagestyle{empty}


\maketitle
\thispagestyle{empty} 


\begin{abstract}
We give the first polynomial-time algorithm for robust regression in the list-decodable setting where an adversary can corrupt a greater than $1/2$ fraction of examples.

For any $\alpha < 1$, our algorithm takes as input a sample $\{(x_i,y_i)\}_{i \leq n}$ of $n$ linear equations where $\alpha n$ of the equations satisfy $y_i = \langle x_i,\ell^*\rangle +\zeta$ for some small noise $\zeta$ and $(1-\alpha)n$ of the equations are {\em arbitrarily} chosen. It outputs a list $L$ of size $O(1/\alpha)$ - a fixed constant - that contains an $\ell$ that is close to $\ell^*$.

Our algorithm succeeds whenever the inliers are chosen from a \emph{certifiably} anti-concentrated distribution $D$. In particular, this gives a $(d/\alpha)^{O(1/\alpha^8)}$ time algorithm to find a $O(1/\alpha)$ size list when the inlier distribution is standard Gaussian. For discrete product distributions that are anti-concentrated only in \emph{regular} directions, we give an algorithm that achieves similar guarantee under the promise that $\ell^*$ has all coordinates of the same magnitude. To  complement our result, we prove that the anti-concentration assumption on the inliers is information-theoretically necessary.

Our algorithm is based on a new framework for list-decodable learning that strengthens the ``identifiability to algorithms'' paradigm based on the sum-of-squares method. 

In an independent and concurrent work, Raghavendra  and Yau~\cite{RY19} also used the Sum-of-Squares method to give a similar result for list-decodable regression.
\end{abstract}

\clearpage


  \microtypesetup{protrusion=false}
  \tableofcontents{}
  \microtypesetup{protrusion=true}

\clearpage

\pagestyle{plain}
\setcounter{page}{1}



\section{Introduction}


In this work, we design algorithms for the problem of linear regression that are robust to training sets with an overwhelming ($\gg 1/2$) fraction of adversarially chosen outliers. 

Outlier-robust learning algorithms have been extensively studied (under the name \emph{robust statistics}) in mathematical statistics~\cite{MR0426989,maronna2006robust,huber2011robust,hampel2011robust}. However, the algorithms resulting from this line of work usually run in time exponential in the dimension of the data~\cite{bernholt2006robust}. An influential line of recent work \cite{journals/jmlr/KlivansLS09,journals/corr/AwasthiBL13, DBLP:journals/corr/DiakonikolasKKL16,DBLP:conf/focs/LaiRV16,DBLP:conf/stoc/CharikarSV17,KothariSteinhardt17,2017KS,HopkinsLi17,DBLP:journals/corr/DiakonikolasKK017a,DBLP:journals/corr/DiakonikolasKS17,DBLP:conf/colt/KlivansKM18} has focused on designing \emph{efficient} algorithms for outlier-robust learning. 

Our work extends this line of research. Our algorithms work in the ``list-decodable learning'' framework. In this model, a majority of the training data (a $1 -\alpha$ fraction) can be adversarially corrupted leaving only an $\alpha \ll 1/2$ fraction of ``inliers''. Since uniquely recovering the underlying parameters is information-theoretically \emph{impossible} in such a setting, the goal is to output a list (with an absolute constant size) of parameters, one of which matches the ground truth. This model was introduced in~\cite{DBLP:conf/stoc/BalcanBV08} to give a discriminative framework for clustering. More recently, beginning with~\cite{DBLP:conf/stoc/CharikarSV17}, various works~\cite{DBLP:conf/stoc/DiakonikolasKS18,KothariSteinhardt17} have considered this as a model of ``untrusted'' data. 

There has been phenomenal progress in developing techniques for outlier-robust learning with a \emph{small} $(\ll 1/2)$-fraction of outliers (e.g. outlier ``filters''~\cite{DBLP:conf/focs/DiakonikolasKK016,DBLP:journals/corr/DiakonikolasKK017a,DBLP:conf/soda/0002D019,DBLP:conf/soda/DiakonikolasKK018}, separation oracles for inliers~\cite{DBLP:conf/focs/DiakonikolasKK016} or the \emph{sum-of-squares} method~\cite{2017KS,HopkinsLi17,KothariSteinhardt17,DBLP:conf/colt/KlivansKM18}). In contrast, progress on algorithms that tolerate the significantly harsher conditions in the list-decodable setting has been slower. The only prior works~\cite{DBLP:conf/stoc/CharikarSV17,DBLP:conf/stoc/DiakonikolasKS18,KothariSteinhardt17} in this direction designed list-decodable algorithms for mean estimation via problem-specific methods. 

In this paper, we develop a principled technique to give the first efficient list-decodable learning algorithm for the fundamental problem of \emph{linear regression}. Our algorithm takes a corrupted set of linear equations with an $\alpha \ll 1/2$ fraction of inliers and outputs a $O(1/\alpha)$-size list of linear functions, one of which is guaranteed to be close to the ground truth (i.e., the linear function that correctly labels the inliers). A key conceptual insight in this result is that list-decodable regression information-theoretically requires the inlier-distribution to be ``anti-concentrated''. Our algorithm succeeds whenever the distribution satisfies a stronger ``certifiable anti-concentration'' condition that is algorithmically ``usable'. This class includes the standard gaussian distribution and more generally, any spherically symmetric distribution with strictly sub-exponential tails.

Prior to our work\footnote{There's a long line of work on robust regression algorithms (see for e.g. \cite{DBLP:conf/nips/Bhatia0KK17,conf/soda/KarmalkarP19}) that can tolerate corruptions only in the \emph{labels}. We are interested in algorithms robust against corruptions in both examples and labels.}, the state-of-the-art outlier-robust algorithms for linear regression~\cite{DBLP:conf/colt/KlivansKM18,conf/soda/DiakonikolasKS19,journals/corr/abs-1803-02815,journals/corr/abs-1802-06485} could handle only a small $(<0.1)$-fraction of outliers even under strong assumptions on the underlying distributions. 

List-decodable regression generalizes the well-studied~\cite{MR1028403,doi:10.1162/neco.1994.6.2.181,MR2757044,2013arXiv1310.3745Y,DBLP:journals/corr/BalakrishnanWY14,DBLP:conf/colt/ChenYC14,DBLP:conf/nips/Zhong0D16,DBLP:conf/aistats/SedghiJA16,DBLP:conf/colt/LiL18} and {\em easier} problem of \emph{mixed linear regression}: given $k$ ``clusters'' of examples that are labeled by one out of $k$ distinct unknown linear functions, find the unknown set of linear functions. All known techniques for the problem rely on faithfully estimating certain \emph{moment tensors} from samples and thus, cannot tolerate the overwhelming fraction of outliers in the list-decodable setting. On the other hand, since we can take any cluster as inliers and treat rest as outliers, our algorithm immediately yields new efficient algorithms for mixed linear regression. Unlike all prior works, our algorithms work without any pairwise separation or bounded condition-number assumptions on the $k$ linear functions.


\paragraph{List-Decodable Learning via the Sum-of-Squares Method} Our algorithm relies on a strengthening of the robust-estimation framework based on the sum-of-squares (SoS) method. This paradigm has been recently used for clustering mixture models~\cite{HopkinsLi17,KothariSteinhardt17} and obtaining algorithms for moment estimation~\cite{2017KS} and linear regression~\cite{DBLP:conf/colt/KlivansKM18} that are resilient to a small $(\ll 1/2)$ fraction of outliers under the mildest known assumptions on the underlying distributions. At the heart of this technique is a reduction of outlier-robust algorithm design to just finding ``simple'' proofs of  unique ``identifiability'' of the unknown parameter of the original distribution from a corrupted sample. However, this principled method works only in the setting with a small ($\ll 1/2$) fraction of outliers. As a consequence, the work of~\cite{KothariSteinhardt17} for mean estimation in the list-decodable setting relied on ``supplementing'' the SoS method with a somewhat problem-dependent technique. 

As an important conceptual contribution, our work yields a framework for list-decodable learning that recovers some of the simplicity of the general blueprint. Central to our framework is a general method of \emph{rounding by votes} for ``pseudo-distributions'' (see Section~\ref{sec:overview})  in the setting with $\gg 1/2$ fraction outliers. Our rounding builds on the work of~\cite{KS19} who developed such a method to give a simpler proof of the list-decodable mean estimation result of~\cite{KothariSteinhardt17}. 

Prior results discussed above hold for any underlying distribution that has upper-bounded low-degree moments and such bounds are  ``captured'' within the SoS system. Such conditions are called as ``certified bounded moment'' inequalities. An important contribution of this work is to formalize \emph{anti-concentration} inequalities within the SoS system and prove such inequalities for natural distribution families. Unlike bounded moment inequalities, there is no canonical encoding within SoS for such statements. We choose an encoding  that allows proving certified anti-concentration for a distribution by showing the existence of a certain approximating polynomial. This allows showing certified anti-concetration via a modular approach relying on a beautiful line of works that construct ``weighted '' polynomial approximators~\cite{2007math......1099L}. 

We believe that our framework for list-decodable estimation and our formulation of certified anti-concentration condition will likely have further applications in outlier-robust learning. 
\subsection{Our Results}
We first define our model for generating samples for list-decodable regression.

\begin{model}[Robust Linear Regression]
For $0 <\alpha < 1$ and $\ell^* \in \R^d$ with $\|\ell^*\|_2 \leq 1$, let $\Lin_D(\alpha,\ell^*)$ denote the following probabilistic process to generate $n$ noisy linear equations $\cS = \{ \langle x_i, a \rangle = y_i\mid 1\leq i \leq n\}$ in variable $a \in \R^d$ with $\alpha n$ \emph{inliers} $\cI$ and $(1-\alpha)n$ \emph{outliers} $\cO$:
\begin{enumerate}
\item Construct $\cI$ by choosing $\alpha n$ i.i.d. samples $x_i \sim D$ and set $y_i = \langle x_i,\ell^* \rangle + \zeta$ for additive noise $\zeta$,
\item Construct $\cO$ by  choosing the remaining $(1-\alpha)n$ equations arbitrarily and potentially adversarially w.r.t the inliers $\cI$.
\end{enumerate}
\label{model:random-equations}
\end{model}
Note that $\alpha$ measures the ``signal'' (fraction of inliers) and can be $\ll 1/2$. The bound on the norm of $\ell^*$ is without any loss of generality. For the sake of exposition, we will restrict to $\zeta = 0$ for most of this paper and discuss (see Remarks~\ref{remark:tolerating-additive-noise-intro} and~\ref{remark:tolerating-additive-noise}) \blfootnote{\textbf{Please note that sections $3$-$6$ are in the supplementary material. }} how our algorithms can tolerate additive noise.

An $\eta$-approximate algorithm for list-decodable regression takes input a sample from $\Lin_D(\alpha,\ell^*)$ and outputs a \emph{constant} (depending only on $\alpha$) size list $L$ of linear functions such that there is some $\ell \in L$ that is $\eta$-close to $\ell^*$.

One of our key conceptual contributions is to identify the strong relationship between \emph{anti-concentration inequalities} and list-decodable regression. Anti-concentration inequalities are well-studied~\cite{ErdosLittlewoodOfford,MR2965282-Tao12,MR2407948-Rudelson08} in probability theory and combinatorics. The simplest of these inequalities upper bound the probability that a high-dimensional random variable has zero projections in any direction.

\begin{definition}[Anti-Concentration]
A $\R^d$-valued zero-mean random variable $Y$ has a $\delta$-\emph{anti-concentrated} distribution if $\Pr[ \iprod{Y,v}=0 ]< \delta$. 
\end{definition}

In Proposition~\ref{prop:identifiability}, we provide a simple but conceptually illuminating proof that anti-concentration is \emph{sufficient} for list-decodable regression. In Theorem~\ref{thm:main-lower-bound}, we prove a sharp converse and show that anti-concentration is information-theoretically \emph{necessary} for even noiseless list-decodable regression. This lower bound surprisingly holds for a natural distribution: uniform distribution on $\zo^d$ and more generally, uniform distribution on $[q]^d$ for  $[q] = \{0,1,2\ldots,q\}$. Our lower bound, in fact, shows the impossibility of even the ``easier'' problem of mixed linear regression on this distribution.

\begin{theorem}[See Proposition~\ref{prop:identifiability} and Theorem~\ref{thm:main-lower-bound}]
There is a (inefficient) list-decodable regression algorithm for $\Lin_D(\alpha,\ell^*)$ with list size $O(\frac{1}{\alpha})$ whenever $D$ is $\alpha$-anti-concentrated. 
Further, there exists a distribution $D$ on $\R^d$ that is $\paren{\alpha+\epsilon}$-anti-concentrated for every $\epsilon >0$ but there is no algorithm for $\frac{\alpha}{2}$-approximate list-decodable regression for $\Lin_D(\alpha,\ell^*)$  that returns a list of size $<d$. \label{thm:lower-bound-results}
\end{theorem}
To handle additive noise of variance $\zeta^2$, we need a control of $\Pr[ |\iprod{x,v}| \leq \zeta] = \E \1(|\iprod{x,v}|\leq \delta)$. 
For our efficient algorithms, in addition, we need that the anti-concentration property
to have a low-degree ``sum-of-squares'' \emph{certificate}. 
SoS is a proof system that reasons about polynomial inequalities. Since the ``core indicator'' $\1(|\iprod{x,v}| \leq \delta)$ is not a polynomial, we phrase certified anti-concentration in terms of an approximating polynomial $p$ for the core indicator. 

For this section, we will use "low-degree sum-of-squares proof" informally and encourage the reader to think of certified anti-concentration as a stronger version of anti-concentration that the SoS method can reason about.

\begin{definition}[Certifiable Anti-Concentration] \label{def:certified-anti-concentration}
A random variable $Y$ has a $k$-\emph{certifiably} $(C,\delta)$-anti-concentrated distribution if there is a univariate polynomial $p$ satisfying $p(0) = 1$ such that there is a degree $k$ sum-of-squares proof of the following two inequalities: 
\begin{enumerate} 
\item $\forall v$, $\langle Y,v\rangle^2 \leq \delta^2 \E \langle Y,v\rangle^2$ implies $(p(\langle Y,v\rangle) -1)^2\leq \delta^2$.
\item $\forall v$, $\|v\|_2^2 \leq 1$ implies  $\E p^2(\left \langle Y,v\rangle \right) \leq C\delta$.
\end{enumerate}
\end{definition} 

\Pnote{seems like it would be nice to have some explanation for the properties of the polynomial asked for in this definition...}
We are now ready to state our main result.
\begin{restatable}[List-Decodable Regression]{theorem}{main} \label{thm:main}
For every $\alpha, \eta > 0$ and a $k$-certifiably $(C,\alpha^2 \eta^2/10C)$-anti-concentrated distribution $D$ on $\R^d$, there exists an algorithm that takes input a sample generated according to $\Lin_D(\alpha,\ell^*)$ and outputs a list $L$ of size $O(1/\alpha)$ such that there is an $\ell \in L$ satisfying $\| \ell - \ell^*\|_2 < \eta$ with probability at least $0.99$ over the draw of the sample. The algorithm needs a sample of size $n = (kd)^{O(k)}$ and runs in time $n^{O(k)} = (kd)^{O(k^2)}$.
\end{restatable} 
\begin{remark}[Tolerating Additive Noise]\label{remark:tolerating-additive-noise-intro}
For additive noise (not necessarily independent across samples) of variance $\zeta^2$ in the inlier labels, our algorithm, in the same running time and sample complexity, outputs a list of size $O(1/\alpha)$ that contains an $\ell$ satisfying $\|\ell-\ell^*\|_2 \leq \frac{\zeta}{\alpha} + \eta$. Since we normalize $\ell^*$ to have unit norm, this guarantee is meaningful only when $\zeta \ll \alpha$. 
\end{remark}

\begin{remark}[Exponential Dependence on $1/\alpha$]
List-decodable regression algorithms immediately yield algorithms for mixed linear regression (MLR) without any assumptions on the components. The state-of-the-art algorithms for MLR with gaussian components~\cite{DBLP:conf/colt/LiL18,DBLP:conf/aistats/SedghiJA16} has an exponential dependence on $k=1/\alpha$ in the running time in the absence of strong separation/condition number assumptions. Liang and Liu~\cite{DBLP:conf/colt/LiL18} (see Page 10 of their paper) use the relationship to learning mixtures of $k$ gaussians (with an $\exp(k)$ lower bound~\cite{DBLP:conf/focs/MoitraV10}) to hint at the impossibility of algorithms with polynomial dependence on $1/\alpha$ for MLR and thus, also for list-decodable regression. 
\end{remark}

\paragraph{Certifiably anti-concentrated distributions} In Section~\ref{sec:certified-anti-concentration}, we show certifiable anti-concentration of some well-studied families of distributions. This includes the standard gaussian distribution and more generally any anti-concentrated spherically symmetric distribution with strictly sub-exponential tails. We also show that simple operations such as scaling, applying well-conditioned linear transformations and sampling preserve certifiable anti-concentration. This yields:
\begin{corollary}[List-Decodable Regression for Gaussian Inliers]
For every $\alpha, \eta > 0$ there's an algorithm for list-decodable regression for the model $\Lin_D(\alpha,\ell^*)$ with $D = \cN(0,\Sigma)$ with $\lambda_{\max}(\Sigma)/\lambda_{min}(\Sigma) = O(1)$ that needs $n = (d/\alpha \eta)^{O\left(\frac{1}{\alpha^4 \eta^4}\right)}$  samples and runs in time $n^{O\left(\frac{1}{\alpha^4 \eta^4}\right)} = (d/\alpha \eta)^{O\left(\frac{1}{\alpha^8 \eta^8}\right)}$.
\end{corollary} 

We note that certifiably anti-concentrated distributions are more restrictive compared to the families of distributions for which the most general robust estimation algorithms work~\cite{2017KS,KothariSteinhardt17,DBLP:conf/colt/KlivansKM18}. To a certain extent, this is inherent. The families of distributions considered in these prior works do not satisfy anti-concentration in general.  And as we discuss in more detail in Section~\ref{sec:overview}, anti-concentration is information-theoretically \emph{necessary} (see Theorem~\ref{thm:lower-bound-results}) for list-decodable regression. This surprisingly rules out families of distributions that might appear natural and ``easy'', for example, the uniform distribution on $\zo^n$.

We rescue this to an extent for the special case when $\ell^*$ in the model $\Lin(\alpha,\ell^*)$ is a "Boolean vector", i.e., has all coordinates of equal magnitude. Intuitively, this helps because  while the the uniform distribution on $\zo^n$ (and more generally, any discrete product distribution) is badly anti-concentrated in sparse directions, they are well anti-concentrated~\cite{ErdosLittlewoodOfford} in the directions that are far from any sparse vectors. 

As before, for obtaining efficient algorithms, we need to work with a \emph{certified} version (see Definition~\ref{def:certified-anti-concentration-Boolean}\blfootnote{\textbf{Please note that sections $3$-$6$ are in the supplementary material. }}) of such a restricted anti-concentration condition. As a specific Corollary (see Theorem~\ref{thm:Booleanmain} for a more general statement), this allows us to show:
\begin{theorem}[List-Decodable Regression for Hypercube Inliers] \label{thm:boolcube}
For every $\alpha, \eta > 0$ there's an $\eta$-approximate algorithm for list-decodable regression for the model $\Lin_D(\alpha,\ell^*)$ with $D$ is uniform on $\zo^d$ that needs $n = (d/\alpha \eta)^{O(\frac{1}{\alpha^4 \eta^4})}$  samples and runs in time $n^{O(\frac{1}{\alpha^4 \eta^4})} = (d/\alpha \eta)^{O(\frac{1}{\alpha^8 \eta^8})}$.
\end{theorem} 

In Section~\ref{sec:hypercube}, we obtain similar results for general product distributions. It is an important open problem to prove certified anti-concentration inequalities for a broader family of distributions.
\paragraph{Concurrent Work}%
In an independent and concurrent work, Raghavendra and Yau obtained similar results for list-decodable linear regression based on the sum-of-squares method~\cite{RY19}.


\section{Overview of our Technique} \label{sec:overview}
In this section, we give a bird's eye view of our approach and illustrate the important ideas in our algorithm for list-decodable regression. 
Thus, given a sample $\cS = \{(x_i,y_i)\}_{i = 1}^n$ from $\Lin_D(\alpha,\ell^*)$, we must construct a constant-size list $L$ of linear functions containing an $\ell$ close to $\ell^*$. 

Our algorithm is based on the sum-of-squares method. We build on the ``identifiability to algorithms'' paradigm developed in several prior works~\cite{DBLP:conf/colt/BarakM16,MR3388192-Barak15,DBLP:conf/focs/MaSS16,2017KS,HopkinsLi17,KothariSteinhardt17,DBLP:conf/colt/KlivansKM18} with some important conceptual differences. 

\paragraph{An \emph{inefficient} algorithm} Let's start by designing an inefficient algorithm for the problem. This may seem simple at the outset. But as we'll see, solving this relaxed problem will rely on some important conceptual ideas that will serve as a starting point for our efficient algorithm. 

Without computational constraints, it is natural to just return the list $L$ of all linear functions $\ell$ that correctly labels all examples in some $S \subseteq \cS$ of size $\alpha n$. We call such an $S$, a large, \emph{soluble} set. True inliers $\cI$ satisfy our search criteria so $\ell^* \in L$. However, it's not hard to show (Proposition~\ref{prop:brute-force-doesn't-work}\blfootnote{\textbf{Please note that sections $3$-$6$ are in the supplementary material. }} ) that one can choose outliers so that the list so generated has size $\exp(d)$ (far from a fixed constant!).

A potential fix is to search instead for a \emph{coarse soluble partition} of $\cS$, if it exists, into disjoint $S_1, S_2,\ldots, S_k$ and  linear functions $\ell_1, \ell_2, \ldots, \ell_k$ so that every $|S_i| \geq \alpha n$ and $\ell_i$ correctly computes the labels in $S_i$. In this setting, our list is small ($k\leq 1/\alpha$). But it is easy to construct samples $\cS$ for which this fails 
because there are coarse soluble partitions of $\cS$ where every $\ell_i$ is far from $\ell^*$. 
\paragraph{Anti-Concentration} 
It turns out that any (even inefficient) algorithm for list-decodable regression provably (see Theorem~\ref{thm:main-lower-bound}) \emph{requires} that the distribution of inliers\footnote{As in the standard robust estimation setting, the outliers are  arbitrary and potentially adversarially chosen.} be sufficiently \emph{anti-concentrated}:


\begin{definition}[Anti-Concentration]
A $\R^d$-valued random variable $Y$ with mean $0$ is $\delta$-anti-concentrated\footnote{Definition~\ref{def:certified-anti-concentration} differs slightly to handle list-decodable regression with additive noise in the inliers.} if for all non-zero $v$, $\Pr[ \iprod{Y,v} = 0 ] < \delta$. A set $T \subseteq \R^d$ is $\delta$-anti-concentrated if the uniform distribution on $T$ is $\delta$-anti-concentrated.
\end{definition}

As we discuss next, anti-concentration is also \emph{sufficient} for list-decodable regression. Intuitively, this is because anti-concentration of the inliers prevents the existence of a soluble set that intersects significantly with $\cI$ and yet can be labeled correctly by $\ell \neq \ell^*$. This is simple to prove in the special case when $\cS$ admits a coarse soluble partition. 


\begin{proposition}
Suppose $\cI$ is $\alpha$-anti-concentrated. Suppose there exists a partition $S_1, S_2,\ldots, S_k \subseteq \cS$ such that each $|S_i| \geq \alpha n$ and there exist $\ell_1, \ell_2, \ldots, \ell_k$ such that $y_j = \iprod{\ell_i, x_j}$ for every $j \in S_i$. Then, there is an $i$ such that $\ell_i = \ell^*$. \label{prop:simple-uniqueness-partition}
\end{proposition}

\begin{proof}
Since $k \leq 1/\alpha$, there is a $j$ such that $|\cI \cap S_j| \geq \alpha |\cI|$. 
Then, $\iprod{x_i, \ell_j}= \iprod{x_i, \ell^*}$ for every $i \in \cI \cap S_j$. 
Thus, $\Pr_{i \sim \cI}[\iprod{x_i,\ell_j-\ell^*} = 0]\geq \alpha$. This contradicts anti-concentration of $\cI$ unless $\ell_j - \ell^* = 0$.
\end{proof}

The above proposition allows us to use \emph{any} soluble partition as a \emph{certificate} of correctness for the associated list $L$. Two aspects of this certificate were crucial in the above argument: 1) \emph{largeness}: each $S_i$ is of size $\alpha n$ - so the generated list is small, and, 2) \emph{uniformity}: every sample is used in exactly one of the sets so $\cI$ must intersect one of the $S_i$s in at least $\alpha$-fraction of the points. 

\paragraph{Identifiability via anti-concentration} For arbitrary $\cS$, a coarse soluble partition might not exist. So we will generalize coarse soluble partitions to obtain certificates that exist for every sample $\cS$ and guarantee largeness and a relaxation of uniformity (formalized below). For this purpose, it is convenient to view such certificates as distributions $\mu$ on $\geq \alpha n$ size soluble subsets of $\cS$ so any collection $\cC\subseteq 2^{\cS}$ of $\alpha n$ size sets corresponds to the uniform distribution $\mu$ on  $\cC$. 

To precisely define uniformity, let $W_i(\mu) = \E_{S \sim \mu} [ \1(i \in S)]$ be the ``frequency of i'', that is, probability that the $i$th sample is chosen to be in a set drawn according to $\mu$. Then, the uniform distribution $\mu$ on any coarse soluble $k$-partition satisfies $W_i = \frac{1}{k}$ for every $i$. That is, all samples $i \in \cS$ are \emph{uniformly} used in such a $\mu$. To generalize this idea, we define $\sum_i W_i(\mu)^2$ as the \emph{distance to uniformity} of $\mu$. Up to a shift, this is simply the variance in the frequencies of the points in $\cS$ used in draws from $\mu$. Our generalization of a coarse soluble partition of $\cS$ is any $\mu$ that minimizes $\sum_i W_i(\mu)^2$, the distance to uniformity, and is thus \emph{maximally uniform} among all distributions supported on large soluble sets. Such a $\mu$ can be found by convex programming. 


 
The following claim generalizes Proposition~\ref{prop:simple-uniqueness-partition} to derive the same conclusion starting from any maximally uniform distribution supported on large soluble sets.

\begin{proposition} \label{prop:fair-weight}
For a maximally uniform $\mu$ on $\alpha n$ size soluble subsets of $\cS$, $\sum_{i \in \cI} \E_{S \sim \mu} [\1 \Paren{i \in S}] \geq \alpha |\cI|$. 
\end{proposition}
The proof proceeds by contradiction (see Lemma~\ref{lem:large-weight-on-inliers}). We show that if $\sum_{i \in \cI} W_i(\mu) \leq \alpha |\cI|$, then we can strictly reduce the distance to uniformity by taking a mixture of $\mu$ with the distribution that places all its probability mass on $\cI$. This allow us to obtain an (inefficient) algorithm for list-decodable regression establishing identifiability. 
\begin{proposition}[Identifiability for List-Decodable Regression] \label{prop:identifiability}
Let $\cS$ be sample from $\Lin(\alpha,\ell^*)$ such that $\cI$ is $\delta$-anti-concentrated for $\delta < \alpha$. Then, there's an (inefficient) algorithm that finds a list $L$ of size $\frac{20}{\alpha-\delta}$ such that $\ell^* \in L$ with probability at least $0.99$.
\end{proposition}
\begin{proof}
Let $\mu$ be \emph{any} maximally uniform distribution over $\alpha n$ size soluble subsets of $\cS$. 
For $k = \frac{20}{\alpha-\delta}$, let $S_1, S_2, \ldots, S_k$ be independent samples from $\mu$.
Output the list $L$ of $k$ linear functions that correctly compute the labels in each $S_i$.

To see why $\ell^* \in L$, observe that $\E |S_j \cap \cI|= \sum_{i \in \cI} \E \1(i \in S_j) \geq \alpha |\cI|$. 
By averaging, $\Pr [|S_j \cap \cI| \geq \frac{\alpha+\delta}{2} |\cI|] \geq \frac{\alpha-\delta}{2}$. Thus, there's  a $j \leq k$ so that $|S_j \cap \cI| \geq \frac{\alpha+\delta}{2} |\cI|$ with probability at least $1-(1-\frac{\alpha-\delta}{2})^{\frac{20}{\alpha-\delta}} \geq 0.99$. We can now repeat the argument in the proof of Proposition~\ref{prop:simple-uniqueness-partition} to conclude that any linear function that correctly labels $S_j$ must equal $\ell^*$.
\end{proof}

\paragraph{An efficient algorithm}
Our identifiability proof suggests the following simple algorithm: 1) find \emph{any} maximally uniform distribution $\mu$ on soluble subsets of size $\alpha n$ of $\cS$, 2) take $O(1/\alpha)$ samples $S_i$ from $\mu$ and 3) return the list of linear functions that correctly label the equations in $S_i$s. This is inefficient because searching over distributions is NP-hard in general. \blfootnote{\textbf{Please note that sections $3$-$6$ are in the supplementary material. }} 

To make this into an efficient algorithm, we start by observing that soluble subsets $S \subseteq \cS$ of size $\alpha n$ can be described by the following set of quadratic equations where $w$ stands for the indicator of $S$ and $\ell$, the linear function that correctly labels the examples in $S$. 

\begin{equation} \label{eq:quadratic-formulation}
  \cA_{w,\ell}\colon
  \left \{
    \begin{aligned}
      &&
      \textstyle\sum_{i=1}^n w_i
      &= \alpha n\\
      &\forall i\in [n].
      & w_i^2
      & =w_i \\
      &\forall i\in [n].
      & w_i \cdot (y_i - \iprod{x_i,\ell})
      & = 0\\
      &
      &\|\ell\|^2
      & \leq 1\\
    \end{aligned}
  \right \}
\end{equation} 

Our efficient algorithm searches for a maximally uniform \emph{pseudo-distribution} on $w$ satisfying \eqref{eq:quadratic-formulation}. Degree $k$ pseudo-distributions (see Section~\ref{sec:preliminaries} for precise definitions) are generalization of distributions that nevertheless ``behave'' just as distributions whenever we take (pseudo)-expectations (denoted by $\pE$) of a class of degree $k$ polynomials. And unlike distributions, degree $k$ pseudo-distributions satisfying\footnote{See Fact~\ref{fact:eff-pseudo-distribution} for a precise statement.} polynomial constraints (such as \eqref{eq:quadratic-formulation}) can be computed in time $n^{O(k)}$. 

For the sake of intuition, it might be helpful to (falsely) think of pseudo-distributions $\tmu$ as simply distributions where we only get access to moments of degree $\leq k$. Thus, we are allowed to compute expectations of all degree $\leq k$ polynomials with respect to $\tmu$. Since $W_i(\tmu) = \pE_{\tmu} w_i$ are just first moments of $\tmu$, our notion of maximally uniform distributions extends naturally to pseudo-distributions. This allows us to prove an analog of Proposition~\ref{prop:fair-weight} for pseudo-distributions and gives us an efficient replacement for Step 1.

\begin{proposition}
For any maximally uniform $\tmu$ of degree $\geq 2$,  $\sum_{i \in \cI} \pE_{\tmu}[w_i]  \geq \alpha |\cI| = \alpha \sum_{i\in [n]} \pE_{\tmu}[w_i]$\mper\label{prop:good-weight-on-inliers}
\end{proposition} 

For Step 2, however, we hit a wall: it's not possible to obtain independent samples from $\tmu$ given only low-degree moments. 

\paragraph{Rounding by Votes} To circumvent this hurdle, our algorithm departs from rounding strategies for pseudo-distributions used in prior works and instead ``rounds'' \emph{each} sample to a candidate linear function. While a priori, this method produces $n$ different candidates instead of one, we will be able to extract a list of $O(\frac{1}{\alpha})$ size that contains the true vector from them. This step will crucially rely on anti-concentration properties of $\cI$. 

Consider the vector $v_i = \frac{\pE_{\tmu}[w_i \ell]}{\pE_{\tmu}[w_i]}$ whenever $\pE_{\tmu}[w_i] \neq 0$ (set $v_i$ to zero, otherwise).  This is simply the (scaled) average, according to $\tmu$, of all the linear functions $\ell$ that are used to label the sets $S$ of size $\alpha n$ in the support of $\tmu$ whenever $i \in S$. Further, $v_i$ depends only on the first two moments of $\tmu$.

We think of $v_i$s as ``votes''%
cast by the $i$th sample for the unknown linear function. 
Let us focus our attention on the votes $v_i$ of $i \in \cI$ - the inliers. We will show that according to the distribution proportional to $\pE[w]$, the average $\ell_2$ distance of $v_i$ from $\ell^*$ is at max $\eta$:

\begin{equation}
\frac{1}{\sum_{i \in \cI} \pE[w_i]} \sum_{i \in \cI} \pE[w_i] \| v_i - \ell^*\|_2  < \eta \mper \label{eq:inliers-guess-well}\tag{$\star$}
\end{equation}

Before diving into \eqref{eq:inliers-guess-well}, let's see how it gives us our efficient list-decodable regression algorithm:

\begin{enumerate}
	\item Find a pseudo-distribution $\tmu$ satisfying \eqref{eq:quadratic-formulation} that minimizes distance to uniformity $\sum_i \pE_{\tmu}[w_i]^2$.
	\item For $O(\frac{1}{\alpha})$ times, independently choose a random index $i \in [n]$ with probability proportional to $\pE_{\tmu}[w_i]$ and return the list of corresponding $v_i$s. 
\end{enumerate} 

Step 1 above is a convex program - it minimizes a norm subject on the convex set of pseudo-distributions - and can be solved in polynomial time. Let's analyze step 2 to see why the algorithm works. Using \eqref{eq:inliers-guess-well} and Markov's inequality, conditioned on $i \in \cI$, $\|v_i - \ell^*\|_2 \leq 2 \eta$ with probability $\geq 1/2$. By Proposition~\ref{prop:good-weight-on-inliers}, $\frac{\sum_{i \in \cI} \pE[w_i]}{\sum_{i \in [n] \pE[w_i]}} \geq \alpha$ so $i \in \cI$ with probability at least $\alpha$. Thus in each iteration of step 2, with probability at least $\alpha/2$, we choose an $i$ such that $v_i$ is $2\eta$-close to $\ell^*$. Repeating $O(1/\alpha)$ times gives us the $0.99$ chance of success.

\paragraph{\eqref{eq:inliers-guess-well} via anti-concentration} As in the information-theoretic argument, \eqref{eq:inliers-guess-well} relies on the anti-concentration of $\cI$.
Let's do a quick proof for the case when $\tmu$ is an actual distribution $\mu$.\blfootnote{\textbf{Please note that sections $3$-$6$ are in the supplementary material. }} 

\begin{proof}[Proof of \eqref{eq:inliers-guess-well} for actual distributions $\mu$]
Observe that $\mu$ is a distribution over $(w,\ell)$ satisfying \eqref{eq:quadratic-formulation}. Recall that $w$ indicates a subset $S \subseteq \cS$ of size $\alpha n$ and $w_i = 1$ iff $i \in S$. And $\ell \in \R^d$ satisfies all the equations in $S$.

By Cauchy-Schwarz, $\sum_i \|\E_\mu[w_i \ell] - \E_\mu[w_i] \ell^*\| \leq \E_{\mu} [\sum_{i \in \cI} w_i \|\ell - \ell^*\|]$. 
Next, as in Proposition~\ref{prop:simple-uniqueness-partition}, since $\cI$ is $\eta$-anti-concentrated, and for all $S$ such that $|\cI \cap S| \geq \eta |\cI|$,  $\ell-\ell^*= 0$. Thus, any such $S$ in the support of $\mu$ contributes $0$ to the expectation above. We will now show that the contribution from the remaining terms is upper bounded by $\eta$. Observe that since $\|\ell-\ell^*\| \leq 2$, \\$\E_\mu [\sum_{i \in \cI} w_i \|\ell - \ell^*\|] = \E_\mu [\1\Paren{|S \cap \cI|< \eta |\cI|}w_i \|\ell - \ell^*\|] = \E_\mu [\sum_{i \in S \cap \cI} \|\ell - \ell^*\|]  \leq 2\eta |\cI|$.
\end{proof}

\paragraph{SoSizing Anti-Concentration} The key to proving \eqref{eq:inliers-guess-well} for pseudo-distributions is a \emph{sum-of-squares} (SoS) proof of anti-concentration inequality: $\Pr_{x \sim \cI} [\iprod{x,v} =0] \leq \eta$ in variable $v$. SoS is a restricted system for proving polynomial inequalities subject to polynomial inequality constraints. Thus, to even ask for a SoS proof we must phrase anti-concentration as a polynomial inequality. 

To do this, let $p(z)$ be a low-degree polynomial approximator for the function $\1\Paren{ z=0}$. 

Then, we can hope to ``replace'' the use of the inequality $\Pr_{x \sim \cI} [\iprod{x,v} =0] \leq \eta \equiv \E_{x \sim \cI} [\1(\iprod{x,v} = 0)] \leq \eta$ in the argument above by $\E_{x \sim \cI}[ p(\iprod{x,v})] \leq \eta$. Since polynomials grow unboundedly for large enough inputs, it is \emph{necessary} for the uniform distribution on $\cI$ to have sufficiently light-tails to ensure that $\E_{x \sim \cI} p(\iprod{x,v})$ is small. In Lemma~\ref{lem:univppty_box}, we show that anti-concentration and strictly sub-exponential tails are \emph{sufficient} to construct such a polynomial.

We can finally ask for a SoS proof for $\E_{x \sim \cI} p(\iprod{x,v}) \leq \eta$ in variable $v$. We prove such \emph{certified} anti-concentration inequalities for broad families of inlier distributions in Section~\ref{sec:certified-anti-concentration}.

\section{Preliminaries}
\label{sec:preliminaries}


In this section, we define pseudo-distributions and sum-of-squares proofs.
See the lecture notes \cite{BarakS16} for more details and the appendix in \cite{DBLP:conf/focs/MaSS16} for proofs of the propositions appearing here.

Let $x = (x_1, x_2, \ldots, x_n)$ be a tuple of $n$ indeterminates and let $\R[x]$ be the set of polynomials with real coefficients and indeterminates $x_1,\ldots,x_n$.
We say that a polynomial $p\in \R[x]$ is a \emph{sum-of-squares (sos)} if there are polynomials $q_1,\ldots,q_r$ such that $p=q_1^2 + \cdots + q_r^2$.

\subsection{Pseudo-distributions}

Pseudo-distributions are generalizations of probability distributions.
We can represent a discrete (i.e., finitely supported) probability distribution over $\R^n$ by its probability mass function $D\from \R^n \to \R$ such that $D \geq 0$ and $\sum_{x \in \mathrm{supp}(D)} D(x) = 1$.
Similarly, we can describe a pseudo-distribution by its mass function by relaxing the constraint $D\ge 0$ to passing certain low-degree non-negativity tests.

Concretely, a \emph{level-$\ell$ pseudo-distribution} is a finitely-supported function $D:\R^n \rightarrow \R$ such that $\sum_{x} D(x) = 1$ and $\sum_{x} D(x) f(x)^2 \geq 0$ for every polynomial $f$ of degree at most $\ell/2$.
(Here, the summations are over the support of $D$.)
A straightforward polynomial-interpolation argument shows that every level-$\infty$-pseudo distribution satisfies $D\ge 0$ and is thus an actual probability distribution.
We define the \emph{pseudo-expectation} of a function $f$ on $\R^d$ with respect to a pseudo-distribution $D$, denoted $\pE_{D(x)} f(x)$, as
\begin{equation}
  \pE_{D(x)} f(x) = \sum_{x} D(x) f(x) \,\mper
\end{equation}
The degree-$\ell$ moment tensor of a pseudo-distribution $D$ is the tensor $\E_{D(x)} (1,x_1, x_2,\ldots, x_n)^{\otimes \ell}$.
In particular, the moment tensor has an entry corresponding to the pseudo-expectation of all monomials of degree at most $\ell$ in $x$.
The set of all degree-$\ell$ moment tensors of probability distribution is a convex set.
Similarly, the set of all degree-$\ell$ moment tensors of degree $d$ pseudo-distributions is also convex.
Unlike moments of distributions, there's an efficient separation oracle for moment tensors of pseudo-distributions.

\begin{fact}[\cite{MR939596-Shor87,parrilo2000structured,MR1748764-Nesterov00,MR1846160-Lasserre01}]
  \label[fact]{fact:sos-separation-efficient}
  For any $n,\ell \in \N$, the following set has a $n^{O(\ell)}$-time weak separation oracle (in the sense of \cite{MR625550-Grotschel81}):
  \begin{equation}
    \Set{ \pE_{D(x)} (1,x_1, x_2, \ldots, x_n)^{\otimes d} \mid \text{ degree-d pseudo-distribution $D$ over $\R^n$}}\,\mper
  \end{equation}
\end{fact}
This fact, together with the equivalence of weak separation and optimization \cite{MR625550-Grotschel81} allows us to efficiently optimize over pseudo-distributions (approximately)---this algorithm is referred to as the sum-of-squares algorithm. The \emph{level-$\ell$ sum-of-squares algorithm} optimizes over the space of all level-$\ell$ pseudo-distributions that satisfy a given set of polynomial constraints (defined below).

\begin{definition}[Constrained pseudo-distributions]
  Let $D$ be a level-$\ell$ pseudo-distribution over $\R^n$.
  Let $\cA = \{f_1\ge 0, f_2\ge 0, \ldots, f_m\ge 0\}$ be a system of $m$ polynomial inequality constraints.
  We say that \emph{$D$ satisfies the system of constraints $\cA$ at degree $r$}, denoted $D \sdtstile{r}{} \cA$, if for every $S\subseteq[m]$ and every sum-of-squares polynomial $h$ with $\deg h + \sum_{i\in S} \max\set{\deg f_i,r}$, $\pE_{D} h \cdot \prod _{i\in S}f_i  \ge 0$.

  We write $D \sdtstile{}{} \cA$ (without specifying the degree) if $D \sdtstile{0}{} \cA$ holds.
  Furthermore, we say that $D\sdtstile{r}{}\cA$ holds \emph{approximately} if the above inequalities are satisfied up to an error of $2^{-n^\ell}\cdot \norm{h}\cdot\prod_{i\in S}\norm{f_i}$, where $\norm{\cdot}$ denotes the Euclidean norm\footnote{The choice of norm is not important here because the factor $2^{-n^\ell}$ swamps the effects of choosing another norm.} of the cofficients of a polynomial in the monomial basis.
\end{definition}

We remark that if $D$ is an actual (discrete) probability distribution, then we have  $D\sdtstile{}{}\cA$ if and only if $D$ is supported on solutions to the constraints $\cA$. We say that a system $\cA$ of polynomial constraints is \emph{explicitly bounded} if it contains a constraint of the form $\{ \|x\|^2 \leq M\}$.
The following fact is a consequence of \cref{fact:sos-separation-efficient} and \cite{MR625550-Grotschel81},

\begin{fact}[Efficient Optimization over Pseudo-distributions]
There exists an $(n+ m)^{O(\ell)} $-time algorithm that, given any explicitly bounded and satisfiable system\footnote{Here, we assume that the bitcomplexity of the constraints in $\cA$ is $(n+m)^{O(1)}$.} $\cA$ of $m$ polynomial constraints in $n$ variables, outputs a level-$\ell$ pseudo-distribution that satisfies $\cA$ approximately. \label{fact:eff-pseudo-distribution}
\end{fact}

\subsection{Sum-of-squares proofs}

Let $f_1, f_2, \ldots, f_r$ and $g$ be multivariate polynomials in $x$.
A \emph{sum-of-squares proof} that the constraints $\{f_1 \geq 0, \ldots, f_m \geq 0\}$ imply the constraint $\{g \geq 0\}$ consists of  polynomials $(p_S)_{S \subseteq [m]}$ such that
\begin{equation}
g = \sum_{S \subseteq [m]} p_S \cdot \Pi_{i \in S} f_i
\mper
\end{equation}
We say that this proof has \emph{degree $\ell$} if for every set $S \subseteq [m]$, the polynomial $p_S \Pi_{i \in S} f_i$ has degree at most $\ell$.
If there is a degree $\ell$ SoS proof that $\{f_i \geq 0 \mid i \leq r\}$ implies $\{g \geq 0\}$, we write:
\begin{equation}
  \{f_i \geq 0 \mid i \leq r\} \sststile{\ell}{}\{g \geq 0\}
  \mper
\end{equation}
For all polynomials $f,g\colon\R^n \to \R$ and for all functions $F\colon \R^n \to \R^m$, $G\colon \R^n \to \R^k$, $H\colon \R^{p} \to \R^n$ such that each of the coordinates of the outputs are polynomials of the inputs, we have the following inference rules:
\begin{align}
&\frac{\cA \sststile{\ell}{} \{f \geq 0, g \geq 0 \} } {\cA \sststile{\ell}{} \{f + g \geq 0\}}, \frac{\cA \sststile{\ell}{} \{f \geq 0\}, \cA \sststile{\ell'}{} \{g \geq 0\}} {\cA \sststile{\ell+\ell'}{} \{f \cdot g \geq 0\}} \tag{addition and multiplication}\\
&\frac{\cA \sststile{\ell}{} \cB, \cB \sststile{\ell'}{} C}{\cA \sststile{\ell \cdot \ell'}{} C}  \tag{transitivity}\\
&\frac{\{F \geq 0\} \sststile{\ell}{} \{G \geq 0\}}{\{F(H) \geq 0\} \sststile{\ell \cdot \deg(H)} {} \{G(H) \geq 0\}} \tag{substitution}\mper
\end{align}
Low-degree sum-of-squares proofs are sound and complete if we take low-level pseudo-distributions as models.
Concretely, sum-of-squares proofs allow us to deduce properties of pseudo-distributions that satisfy some constraints.
\begin{fact}[Soundness]
  \label{fact:sos-soundness}
  If $D \sdtstile{r}{} \cA$ for a level-$\ell$ pseudo-distribution $D$ and there exists a sum-of-squares proof $\cA \sststile{r'}{} \cB$, then $D \sdtstile{r\cdot r'+r'}{} \cB$.
\end{fact}
If the pseudo-distribution $D$ satisfies $\cA$ only approximately, soundness continues to hold if we require an upper bound on the bit-complexity of the sum-of-squares $\cA \sststile{r'}{} B$  (number of bits required to write down the proof). In our applications, the bit complexity of all sum of squares proofs will be $n^{O(\ell)}$ (assuming that all numbers in the input have bit complexity $n^{O(1)}$). This bound suffices in order to argue about pseudo-distributions that satisfy polynomial constraints approximately.

The following fact shows that every property of low-level pseudo-distributions can be derived by low-degree sum-of-squares proofs.
\begin{fact}[Completeness]
  \label{fact:sos-completeness}
  Suppose $d \geq r' \geq r$ and $\cA$ is a collection of polynomial constraints with degree at most $r$, and $\cA \vdash \{ \sum_{i = 1}^n x_i^2 \leq B\}$ for some finite $B$.

  Let $\{g \geq 0 \}$ be a polynomial constraint.
  If every degree-$d$ pseudo-distribution that satisfies $D \sdtstile{r}{} \cA$ also satisfies $D \sdtstile{r'}{} \{g \geq 0 \}$, then for every $\epsilon > 0$, there is a sum-of-squares proof $\cA \sststile{d}{} \{g \geq - \epsilon \}$.
\end{fact}
We will use the following Cauchy-Schwarz inequality for pseudo-distributions:
\begin{fact}[Cauchy-Schwarz for Pseudo-distributions]
Let $f,g$ be polynomials of degree at most $d$ in indeterminate $x \in \R^d$. Then, for any degree d pseudo-distribution $\tmu$,
$\pE_{\tmu}[fg] \leq \sqrt{\pE_{\tmu}[f^2]} \sqrt{\pE_{\tmu}[g^2]}$.
 \label{fact:pseudo-expectation-cauchy-schwarz}
\end{fact} 
The following fact is a simple corollary of the fundamental theorem of algebra:
\begin{fact}
For any univariate degree $d$ polynomial $p(x) \geq 0$ for all $x \in \R$, 
$\sststile{d}{x} \Set{p(x) \geq 0}$.
 \label{fact:univariate}
\end{fact}
This can be extended to univariate polynomial inequalities over intervals of $\R$. 
\begin{fact}[Fekete and Markov-Lukács, see \cite{laurent2009sums}]
For any univariate degree $d$ polynomial $p(x) \geq 0$ for $x \in [a, b]$,  $\Set{x\geq a, x \leq b} \sststile{d}{x} \Set{p(x) \geq 0}$.  \label{fact:univariate-interval}
\end{fact}


\section{Algorithm for List-Decodable Robust Regression}
In this section, we describe and analyze our algorithm for list-decodable regression and prove our first main result restated here.
\main*
We will analyze Algorithm~\ref{alg:noisy-regression-gaussian} to prove Theorem~\ref{thm:main}.
\begin{equation}
  \cA_{w,\ell}\colon
  \left \{
    \begin{aligned}
      &&
      \textstyle\sum_{i=1}^n w_i
      &= \alpha n\\
      &\forall i\in [n].
      & w_i^2
      & =w_i \\
      &\forall i\in [n].
      & w_i \cdot (y_i - \iprod{x_i,\ell})
      & = 0\\
      &
      &\sum_{i \leq d} \ell_i^2 \leq 1\\
    \end{aligned}
  \right \}
\end{equation} 
\begin{mdframed}
  \begin{algorithms}[List-Decodable Regression]
    \label[algorithm]{alg:noisy-regression-gaussian}\mbox{}
    \begin{description}
    \item[Given:]
    Sample $\cS$ of size $n$ drawn according to $\Lin(\alpha,n,\ell^*)$ with inliers $\cI$, $\eta > 0$. 
    \item[Output:]
    	A list $L \subseteq \R^d$ of size $O(1/\alpha)$ such that there exists a $\ell \in L$ satisfying $\|\ell -\ell^*\|_2 < \eta$.
    \item[Operation:]\mbox{}
    \begin{enumerate}
		\item Find a degree $O(1/\alpha^4\eta^4)$ pseudo-distribution $\tilde{\mu}$ satisfying $\cA_{w,\ell}$ that minimizes $\|\pE[w]\|_2$.
		\item For each $i \in [n]$ such that $\pE_{\tmu}[w_i] > 0$, let $v_i = \frac{\pE_{\tmu}[w_i \ell]}{\pE_{\tmu}[w_i]}$. Otherwise, set $v_i =0$.
		\item 
		Take $J$ be a random multiset formed by union of $O(1/\alpha)$ independent draws of $i \in [n]$ with probability $\frac{\pE[w_i]}{\alpha n}$.
		\item Output $L = \{v_i \mid i \in J\}$ where $J \subseteq [n]$.
	\end{enumerate}
    \end{description}    
  \end{algorithms}
\end{mdframed}

Our analysis follows the discussion in the overview. 
We start by formally proving \eqref{eq:inliers-guess-well}. 

\begin{lemma}
For any $t \geq k$ and  any $\cS$ so that $\cI \subseteq \cS$ is $k$-certifiably $(C,\alpha^2\eta^2/4C)$-anti-concentrated,
\[
\cA_{w,\ell} \sststile{t}{w,\ell} \Set{\frac{1}{|\cI|}\sum_{i \in \cI}^n w_i \|\ell - \ell^*\|^2_2\leq \frac{\alpha^2\eta^2}{4}}
\]
\label{lem:close-on-inliers}
\end{lemma}

\begin{proof}
We start by observing:
$\cA_{w,\ell} \sststile{2}{\ell} \|\ell - \ell^*\|_2^2 \leq 2$.

Since $\cI$ is $(C,\alpha \eta/2C)$-anti-concentrated, there exists a univariate polynomial $p$ such that $\forall i$:

\begin{equation}
\Set{w_i\iprod{x,\ell-\ell^*} = 0} \sststile{\ell}{k} \Set{p(w_i\iprod{x_i,\ell-\ell^*}) = 1}\mcom \label{eq:value-at-0}
\end{equation}
and 

\begin{equation} \label{eq:sos-expectation-p}
\Set{\|\ell\|^2 \leq 1} \sststile{\ell}{k} \Set{\frac{1}{|\cI|} \sum_{i \in \cI} p( \iprod{x_i,\ell-\ell^*})^2 \leq \frac{\alpha^2  \eta^2}{4}}\mper
\end{equation}
Using \eqref{eq:value-at-0}, we have:

\begin{align*}
\cA_{w,\ell} \sststile{t+2}{w,\ell} \Set{ 1-p^2 (w_i \langle x_i, \ell - \ell^* \rangle) = 0} \sststile{t+2}{w,\ell} \Set{1- w_i p^2 (\langle x_i, \ell - \ell^* \rangle) = 0}\mper\\
\end{align*}

Using \eqref{eq:sos-expectation-p} and $\cA_{w,\ell} \sststile{w}{2} \Set{w_i^2 = w_i}$, we thus have:
\begin{align*}
\cA_{w,\ell}  \sststile{t+2}{w,\ell} \Bigl\{ \frac{1}{|\cI|} \sum_{i \in \cI} w_i \|\ell-\ell^*\|_2^2 &= \frac{1}{|\cI|} \sum_{i \in \cI} w_i \|\ell-\ell^*\|_2^2 w_i p^2 (\langle x_i, \ell-\ell^* \rangle) 
= \frac{1}{|\cI|} \sum_{i \in \cI} w_i \|\ell-\ell^*\|_2^2 p^2 (\langle x_i, \ell-\ell^*\rangle) \\
&\leq \frac{1}{|\cI|} \sum_{i \in \cI} \|\ell-\ell^*\|_2^2 p^2 (\langle x_i, \ell-\ell^* \rangle)  \leq \frac{\alpha^2 \eta^2}{4} 
\Bigr\}\mper
\end{align*}

\end{proof}

As a consequence of this lemma, we can show that a constant fraction of the $v_i$ for $i \in \cI$ constructed in the algorithm are close to $\ell^*$. 

\begin{lemma}
For any $\tmu$ of degree $k$ satisfying $\cA_{w,\ell}$, 
$\frac{1}{|\cI|} \sum_{i \in \cI} \pE[w_i] \cdot \|v_i -\ell^*\|_2 \leq \frac{\alpha}{2}\eta$.
\label{lem:votes-are-close}
\end{lemma}
\begin{proof}
By Lemma~\ref{lem:close-on-inliers}, we have:
$
\cA_{w,\ell} \sststile{k}{w,\ell} \Set{\frac{1}{|\cI|}\sum_{i \in \cI}^n w_i \|\ell - \ell^*\|^2_2\leq \frac{\alpha^2 \eta^2}{4} }
$.

We also have: $\cA_{w,\ell} \sststile{2}{w,\ell} \Set{w_i^2 - w_i =0}$ for any $i$. This yields:
\[
\cA_{w,\ell} \sststile{k}{w,\ell} \Set{\frac{1}{|\cI|}\sum_{i \in \cI}^n \|w_i\ell - w_i\ell^*\|^2_2\leq\frac{\alpha^2 \eta^2}{4} }
\]

Since $\tmu$ satisfies $\cA_{w,\ell}$, taking pseudo-expectations yields:$
 \frac{1}{\cI} \sum_{i \in \cI}\pE \| w_i\ell  - w_i \ell^*\|_2^2 \leq \frac{\alpha^2 \eta^2}{4}$.

By Cauchy-Schwarz for pseudo-distributions (Fact~\ref{fact:pseudo-expectation-cauchy-schwarz}), we have:
\[
 \Paren{\frac{1}{\cI} \sum_{i \in \cI} \| \pE[w_i\ell]  - \pE[w_i] \ell^*\|_2}^2 \leq \frac{1}{\cI} \sum_{i \in \cI} \| \pE[w_i\ell]  - \pE[w_i] \ell^*\|_2^2 \leq \frac{\alpha^2 \eta^2}{4} \mper
\]

Using $v_i = \frac{\pE[w_i \ell]}{\pE[w_i]}$ if $\pE[w_i] >0$ and $0$ otherwise, we have:
$ \frac{1}{\cI} \sum_{i \in \cI, \pE[w_i] > 0} \pE[w_i]\cdot \|v_i -  \ell^*\|_2 \leq \frac{\alpha}{2}\eta $.

\end{proof}

Next, we formally prove that maximally uniform pseudo-distributions satisfy Proposition~\ref{prop:good-weight-on-inliers}.
\begin{lemma}
For any $\tilde{\mu}$ of degree $\geq 4$ satisfying $\cA_{w,\ell}$ that minimizes $\|\pE[w]\|_2$, $\sum_{i \in \cI} \pE_{\tilde{\mu}}[w_i] \geq \alpha^2 n$.
 \label{lem:large-weight-on-inliers}
\end{lemma}

\begin{proof}
Let $u = \frac{1}{\alpha n}\pE[w]$. Then, $u$ is a non-negative vector satisfying $\sum_{i \sim [n]} u_i = 1$.


Let $\wt(\cI) = \sum_{i \in \cI} u_i$ and $\wt(\cO) = \sum_{i \not \in \cI} u_i$. Then, $\wt(\cI) + \wt(\cO) = 1$.

We will show that if $\wt(\cI) < \alpha$, then there's a pseudo-distribution $\tmu'$ that satisfies $\cA_{w,\ell}$ and has a lower value of $\|\pE[w]\|_2$. This is enough to complete the proof. 

To show this, we will ``mix'' $\tmu$ with another pseudo-distribution satisfying $\cA_{w,\ell}$. Let $\tmu^*$ be the \emph{actual} distribution supported on single $(w,\ell)$ - the indicator $\1_{\cI}$ and $\ell^*$. Thus, $\pE_{\tmu^*} w_i = 1$ iff $i \in \cI$ and $0$ otherwise. $\tmu^*$ clearly satisfies $\cA_{w,\ell}$. Thus, any convex combination (mixture) of $\tmu$ and $\tmu^*$ also satisfies $\cA_{w,\ell}$. 

Let $\tmu_{\lambda} = (1-\lambda) \tmu + \lambda \tmu^*$. We will show that there is a $\lambda >0$ such that $\|\pE_{\tmu_{\lambda}}[w]\|_2 < \|\pE[w]\|_2$.

We first lower bound $\|u\|_2^2$ in terms of $\wt(\cI)$ and $\wt(\cO)$. Observe that for any fixed values of $\wt(\cI)$ and $\wt(\cO)$, the minimum is attained by the vector $u$ that ensures $u_i = \frac{1}{\alpha n} \wt(\cI)$ for each $i \in \cI$ and $u_i = \frac{1}{(1-\alpha)n} \wt(\cO)$. \begin{align*}
  \text{This gives }  \|u\|^2 &\geq \left( \frac{\wt(\cI)}{\alpha n} \right)^2 \alpha n + \left(\frac{1-\wt(\cI)}{(1-\alpha) n}\right)^2 (1-\alpha) n 
    = \frac{1}{\alpha n}\cdot \left( \wt(\cI) + (1-\wt(\cI))^2 \left(\frac{\alpha}{1-\alpha}\right) \right)\mper 
\end{align*}
Next, we compute the the $\ell_2$ norm of $u' = \frac{1}{\alpha n} \pE_{\tmu_{\lambda}} w$ as:
\[ \|u'\|_2^2 = (1-\lambda)^2 \|u\|^2 +  \frac{\lambda^2}{\alpha n} + 2 \lambda(1-\lambda)\frac{\wt(\cI)}{\alpha n} \mper\] 
\begin{align*}
    \text{Thus, } \|u'\|^2 - \|u\|^2  &= (-2\lambda+\lambda^2) \|u\|^2 +  \frac{\lambda^2}{\alpha n} + 2 \lambda(1-\lambda)\frac{\wt(\cI)}{\alpha n}\\
    &\leq \frac{-2\lambda+\lambda^2}{\alpha n}\cdot \left( \wt(\cI)^2 + (1-\wt(\cI))^2 \frac{\alpha}{1-\alpha} \right) +  \frac{\lambda^2}{\alpha n} + 2 \lambda(1-\lambda)\frac{\wt(\cI)}{\alpha n}
\end{align*}
\begin{align*} \text{Rearranging, }\|u\|^2 - \|u'\|^2 &\geq \frac{\lambda}{\alpha n} \left( (2-\lambda) \cdot \left( \wt(\cI)^2 + (1-\wt(\cI))^2 \left(\frac{\alpha}{1-\alpha}\right) \right) -  \lambda - 2 (1-\lambda)\wt(\cI)\right)\\
&\geq \frac{\lambda ( 2-\lambda)}{\alpha n} \left( \wt(\cI)^2 + (1-\wt(\cI))^2 \frac{\alpha}{1-\alpha} - \wt(\cI)\right)
\end{align*}
Now, whenever $\wt(\cI) < \alpha$, $\wt(\cI)^2 + (1-\wt(\cI))^2 \frac{\alpha}{1-\alpha} - \wt(\cI)> 0$. Thus, we can choose a small enough $\lambda > 0$ so that $\|u\|^2 - \|u'\|^2 > 0$. 

\end{proof}

Lemma~\ref{lem:large-weight-on-inliers} and Lemma~\ref{lem:votes-are-close} immediately imply the correctness of our algorithm. 
\begin{proof}[Proof of Main Theorem~\ref{thm:main}]
First, since $D$ is $k$-certifiably $(C,\alpha \eta/4C)$-anti-concentrated, Lemma~\ref{lem:sampling-preserves-certified-anti-concentrated} implies taking $\geq n = (kd)^{O(k)}$ samples ensures that $\cI$ is $k$-certifiably $(C,\alpha \eta/2C)$-anti-concentrated with probability at least $1-1/d$. Let's condition on this event in the following. 

Let $\tmu$ be a pseudo-distribution of degree $t$ satisfying $\cA_{w,\ell}$ and minimizing $\|\pE[w]\|_2$.
Such a pseudo-distribution exists as can be seen by just taking the distribution with a single-point support $w$ where $w_i = 1$ iff $i \in \cI$. 

From Lemma~\ref{lem:votes-are-close}, we have: 
$
\frac{1}{|\cI|} \sum_{i \in \cI} \pE[w_i] \cdot \|v_i -\ell^*\|_2 \leq \frac{\alpha}{2} \eta 
$. Let $Z = \frac{1}{\alpha n} \sum_{i \in \cI} \pE[w_i]$. By a rescaling, we obtain:
\begin{equation} 
\frac{1}{|\cI|} \sum_{i \in \cI} \frac{\pE[w_i]}{Z} \cdot \|v_i -\ell^*\|_2 \leq \frac{1}{Z} \frac{\alpha}{2} \eta\mper
\end{equation}
Using Lemma~\ref{lem:large-weight-on-inliers}, $Z \geq \alpha$. Thus, 
\begin{equation} 
\label{eq:good-on-average}
\frac{1}{|\cI|} \sum_{i \in \cI} \frac{\pE[w_i]}{Z} \cdot \|v_i -\ell^*\|_2 \leq \eta/2 \mper
\end{equation}

Let $i \in [n]$ be chosen with probability $\frac{\pE[w_i]}{\alpha n}$. 
Then, $i \in \cI$ with probability $Z \geq \alpha$. 
By Markov's inequality applied to \eqref{eq:good-on-average}, with $\frac{1}{2}$ conditioned on $i \in \cI$, $\|v_i - \ell^*\|_2 < \eta$. Thus, in total, with probability at least $\alpha/2$, $\|v_i - \ell^*\|_2 \leq \eta$.
Thus, the with probability at least $0.99$ over the draw of the random set $J$, the list constructed by the algorithm contains an $\ell$ such that $\|\ell - \ell^*\|_2 \leq \eta$.

Let us now account for the running time and sample complexity of the algorithm.
The sample size for the algorithm is dictated by Lemma~\ref{lem:sampling-preserves-certified-anti-concentrated} and is $(kd)^{O(k)}$, which for our choice of $p$ goes as $(kd)^{O(k)}$.
A pseudo-distribution satisfying $\cA_{w,\ell}$ and minimizing $\|\pE[w]\|_2$ can be found in time $n^{O(k)} = (kd)^{O(k^2)}$. 
The rounding procedure runs in time at most $O(nd)$. 
\end{proof}
\begin{remark}[Tolerating Additive Noise] \label{remark:tolerating-additive-noise}
To tolerate independent additive noise, our algorithm and analysis change minimally. For an additive noise of variance $\zeta^2 \ll \alpha^2  \eta^2$ in the inliers, we modify $\cA_{w,\ell}$ by replacing the constraint $\forall i$, $w_i \cdot (y_i - \iprod{x_i,\ell})= 0$ by $\forall i$, $\pm w_i \cdot (y_i - \iprod{x_i, \ell}) \leq 4\zeta$. And $\sum_{i = 1}^n w_i = \alpha n$ to $\sum_{i = 1}^n w_i = (\alpha/2) n$. 

This means that instead of searching for a subsample of size $\alpha n$ that has a exact solution $\ell$, we search for a subsample of size $\alpha/2 n$ where there's a solution $\ell$ with an additive error of at most $2\zeta$. With additive noise of variance $\zeta^2$, it is easy to check that there's a subset of $1/2$ fraction of inliers that satisfies this property. Thus, $\cA_{w,\ell}$ is feasible. 

Our analysis remains exactly the same except for one change in the proof of Lemma~\ref{lem:close-on-inliers}. We start from a distribution that is $(C,\alpha \eta \zeta/100C)$-certifiably anti-concentrated. And instead of inferring that $p\paren{w_i(y_i - \iprod{x_i,\ell})} = 1$, we use that whenever $\pm \paren{y_i - \iprod{x_i,\ell}} \leq 4\zeta$, $p^2(\paren{y_i - \iprod{x_i,\ell}}) \geq 1-4\zeta$. 
\end{remark}

\subsection{List-Decodable Regression for Boolean Vectors} \label{sec:hypercube}

In this section, we show algorithms for list-decodable regression when the distribution on the inliers satisfies a weaker anti-concentration condition. This allows us to handle more general inlier distributions including the product distributions on $\on^d$, $[0,1]^d$ and more generally any product domain. We however require that the unknown linear function be ``Boolean'', that is, all its coordinates be of equal magnitude.

We start by defining the weaker anti-concentration inequality. Observe that if $v \in \R^d$ satisfies $v_i^3 = \frac{1}{d} v_i$ for every $i$, then the coordinates of $v$ are in $\{0,\pm \frac{1}{\sqrt{d}}\}$.

\begin{definition}[Certifiable Anti-Concentration for Boolean Vectors] \label{def:certified-anti-concentration-Boolean}
A $\R^d$ valued random variable $Y$ is $k$-\emph{certifiably} $(C,\delta)$-anti-concentrated in \emph{Boolean directions} if there is a univariate polynomial $p$ satisfying $p(0) = 1$ such that there is a degree $k$ sum-of-squares proof of the following two inequalities: for all $x^2 \leq \delta^2$, $(p(x) - 1)^2 \leq \delta^2$ and for all $v$ such that $v_i^3 = \frac{4}{d} v_i$ for all $i$, $\|v\|^2 \E_{Y} p(\iprod{Y,v})^2 \leq C\delta$. 
\end{definition} 

We can now state the main result of this section.

\begin{theorem}[List-Decodable Regression in Boolean Directions]
For every $\alpha, \eta$, there's a algorithm that takes input a sample generated according to $\Lin_D(\alpha,n,\ell^*)$ in $\R^d$ for $D$ that is $k$-certifiably $(C,\alpha \eta/10C)$-anti-concentrated in Boolean directions and $\ell^* \in \Set{\pm{\frac{1}{\sqrt{d}}}}^d$ and outputs a list $L$ of size $O(1/\alpha)$ such that there's an $\ell \in L$ satisfying $\|\ell-\ell^*\| <\eta$ with probability at least $0.99$ over the draw of the sample. The algorithm requires a sample of size $n \geq (d/\alpha \eta)^{O(\frac{1}{\alpha^2 \eta^2})}$ and runs in time $n^{O(k)} = (d/\alpha\eta)^{O(k^2)}$. \label{thm:Booleanmain}
\end{theorem}

The only  difference in our algorithm and rounding is that instead of the constraint set $\cA_{w,\ell}$, we will work with $\cB_{w,\ell}$  that has an additional constraint $\ell_i^2 = \frac{1}{d}$ for every $i$. Our algorithm is exactly the same as Algorithm~\ref{alg:noisy-regression-gaussian} replacing $\cA_{w,\ell}$ by $\cB_{w,\ell}$.

\begin{equation}
  \cB_{w,\ell}\colon
  \left \{
    \begin{aligned}
      &&
      \textstyle\sum_{i=1}^n w_i
      &= \alpha n\\
      &\forall i\in [n],
      & w_i^2
      & =w_i \\
      &\forall i\in [n],
      & w_i \cdot (y_i - \iprod{x_i,\ell})
      & = 0\\
      &\forall in \in [d],
      &\ell_i^2 &= \frac{1}{d}\\
    \end{aligned}
  \right \}
\end{equation} 

We will use the following fact in our proof of Theorem~\ref{thm:Booleanmain}.

\begin{lemma}
If $a,b$ satisfy $a^2 = b^2 = \frac{2}{d}$, then, 
$
(a- b)^{3} = \frac{1}{d} (a - b)
$ \label{lem:simple-cubic}
\end{lemma}
\begin{proof}
$(a-b)^3 = a^3 - b^3 - 3a^2 b + 3ab^2 = \frac{1}{d}(a-b-3b+3a) = \frac{4}{d}(a-b)$.
\end{proof}

\begin{proof}[Proof of Theorem~\ref{thm:Booleanmain}]
The proof remains the same as in the previous section with one additional step. 
First, we can obtain the analog of Lemma~\ref{lem:close-on-inliers} with a few quick modifications to the proof. 
Then, Lemma~\ref{lem:votes-are-close} follows from modified Lemma~\ref{lem:close-on-inliers} as in the previous section. 
And the proof of Lemma~\ref{lem:large-weight-on-inliers} remains exactly the same. 
We can then put the above lemmas together just as in the proof of Theorem~\ref{thm:main}. 

We now describe the modifications to obtain the analog of Lemma~\ref{lem:close-on-inliers}. 
The key additional step in the proof of the analog of Lemma~\ref{lem:close-on-inliers} which follows immediately from Lemma~\ref{lem:simple-cubic}.

\[
\Set{\forall i\text{ } \ell_i^2 = \frac{1}{d}} \sststile{4}{\ell} \Set{ (\ell_i - \ell_i^*)^{3} = \frac{4}{d} (\ell_i - \ell_i^*)}
\]

This allows us to replace the usage of certifiable anti-concentration by certifiable anti-concentration for Boolean vectors and derive:
\[
\Set{\forall i\text{ } \ell_i^2 = \frac{2}{d}} \sststile{4}{\ell} \Set{\frac{1}{|\cI|}\sum_{i \in \cI} p(\iprod{x_i,\ell-\ell^*})^2 \leq \frac{\alpha^2 \eta^2}{4} }
\]

The rest of the proof of Lemma~\ref{lem:close-on-inliers} remains the same.

\end{proof}

\section{Certifiably Anti-Concentrated Distributions} \label{sec:certified-anti-concentration}
In this section, we prove certifiable anti-concentration inequalities for some basic families of distributions. We first formally state the definition of certified-anti-concentration.

\begin{definition}[Certifiable Anti-Concentration] \label{def:formal-certified-anti-concentration}
A $\R^d$-valued zero-mean random variable $Y$ has a $(C,\delta)$-\emph{anti-concentrated} distribution if $\Pr[ |\iprod{Y,v}| \leq \delta \sqrt{ \E \iprod{Y,v}^2}]\leq C\delta$.

$Y$ has a $k$-\emph{certifiably} $(C,\delta)$-anti-concentrated distribution if there is a univariate polynomial $p$ satisfying $p(0) = 1$ such that
\begin{enumerate} 
\item $\left \{ \langle Y,v\rangle^2 \leq \delta^2 \E \langle Y,v\rangle^2 \right \}  \sststile{k}{v} \left \{ (p(\langle Y,v\rangle) -1)^2\leq \delta^2 \right \}$.
\item $\left \{ \|v\|_2^2 \leq 1 \right \} \sststile{k}{v} \left \{ \|v\|_2^2 \E p^2(\left \langle Y,v\rangle \right) \leq C\delta \right \}$.
\end{enumerate}
We will say that such a polynomial $p$ ``witnesses the certifiable anti-concentration of $Y$''. We will use the phrases ``$Y$ has a certifiably anti-concentrated distribution'' and ``$Y$ is a certifiably anti-concentrated random variable'' interchangeably. 
\end{definition} 
As one would like, the definition above is scale invariant:
\begin{lemma}[Scale invariance] \label{lem:shift-scale-invariance}
Let $Y$ be a $k$-certifiably $(C,\delta)$-anti-concentrated random variable. Then, so is $cY$ for any $c\neq 0$. 
\end{lemma}
\begin{proof}
Let $p$ be the polynomial that witnesses the certifiable anti-concentration of $Y$. 
Then, observe that $q(z) = p(z/c)$ satisfies the requirements of the definition for $cY$. 
\end{proof}

\begin{lemma}[Certified anti-concentration of gaussians]\label{lem:spherically-symmetric-certifiable-gaussians}
For every $0.1 > \delta >0$, there is a $k = O\left(\frac{\log^2(1/\delta)}{\delta^2}\right)$  such that $\cN(0,I)$ is $k$-certifiably $(2,2\delta)$-anti-concentrated.
\end{lemma}
\begin{proof}
Lemma~\ref{lem:univppty_box} yields that there exists an univariate even polynomial $p$ of degree $k$ as above such that for all $v$, whenever $|\iprod{x,v}| \leq \delta$, $p(\iprod{x,v}) \leq 2\delta$, and whenever $\|v\|^2 \leq 1$, $\E_{x \sim \cN(0,I)} p(  \iprod{x,v})^2 \leq 2\delta$. Since $p$ is even, $p(z) = \frac{1}{2}(p(z) +  p(-z))$  and thus, any monomial in $p(z)$ with non-zero  coefficient must be of even degree. Thus, $p(z) =  q(z^2)$ for some polynomial $q$ of degree $k/2$. 

The first property above for $p$ implies that whenever $z \in [0,\delta]$, $p(z) \leq 2 \delta$. 
By Fact~\ref{fact:univariate-interval}, we obtain that:
$\Set{\iprod{x,v}^2 \leq \delta^2} \sststile{k}{v} \Set{p(\iprod{x,v})^2 \leq \delta }$.
Next, observe that for any $j$, $\E_{x \sim \cN(0,I)} \iprod{x,v}^{2j} = (2j)!! \cdot \|v\|_2^{2j}$.
Thus, $\|v\|_2^2 \E_{x \sim \cN(0,I)} p^2(\iprod{x,v})$ is a univariate polynomial $F$ in $\|v\|_2^2$. 
The second property above thus implies that $F(\|v\|_2^2) \leq C\delta$ whenever $\|v\|_2^2 \leq 1$.
By another application of Fact~\ref{fact:univariate-interval}, we obtain:
$\Set{\|v\|_2^2 \leq 1} \sststile{k}{v} \Set{\E_{x\sim \cN(0,I)} p(\iprod{x,v})^2 \leq 2\delta}$. 
\end{proof}

We say that $Y$ is a \emph{spherically symmetric} random variable over $\R^d$ if for every orthogonal matrix $R$, $RY$ has the same distribution as $Y$. Examples include the standard gaussian random variable and uniform (Haar) distribution on $\S^{d-1}$. Our argument above for the case of standard gaussian extends to any distribution that is spherically symmetric and has sufficiently light tails. 

\begin{lemma}[Certified anti-concentration of spherically symmetric, light-tail distributions] \label{lem:spherically-symmetric-certifiable-anti-concentration}
Suppose $Y$ is a $\R^d$-valued, spherically symmetric random variable  such that for any $k \in (0, 2)$, for all $t$ and for all $v$, 
$\Pr[\langle v,  Y\rangle \geq t \sqrt{\E \langle Y, v \rangle ^2}] \leq Ce^{-t^{2/k}/C}$ and for all $\eta > 0$, $\Pr_{x \sim D} [ |x| < \eta\sigma] \leq C\eta$, for some absolute constant $C >0$. Then, for $d = O\left(\frac{\log^{(4+k)/(2-k)}(1/\delta)}{\delta^{2/(2-k)}}\right)$, $Y$ is $d$-certifiably $(10C, \delta)$-anti-concentrated.
\end{lemma}
\Pnote{Good to see if $C''$ is just $CC'$ or something.}

\begin{lemma}[Certified anti-concentration under sampling] \label{lem:sampling-preserves-certified-anti-concentrated}
Let $D$ be $k$-certifiably $(C,\delta)$-anti-concentrated, subexponential and unit covariance distribution. Let $S$ be a collection of $n$ independent samples from $D$. Then, for $n \geq \Omega \left( (kd\log(d))^{O(k)} \right)$, with probability at least $1-1/d$, the uniform distribution on $S$ is $(2C,\delta)$-anti-concentrated. 
\end{lemma}
\begin{proof} 
Let $p$ be the degree $k$ polynomial that witnesses the certifiable anti-concentration of $D$. 
Let $Y$ be the random variable with distribution $D'$, the uniform distribution on $n$ i.i.d. samples from $D$.
We will show that $p$ also witnesses that $k$-certifiable $(4C,\delta/2)$-anti-concentration of $Y$. To this end it is sufficient to take enough samples such that the following holds. 
$\Pr\left( \left| \E_{D} [p^2(\langle Y, v\rangle)] - \E_{D'} [p^2(\langle Y, v\rangle)] \right| > \E_{D} [p^2(\langle Y, v\rangle)]/2 \right) < 1/d$. 
Observe that $p^2(\langle Y, v\rangle)$ may be written as $\langle c(Y)c(Y)^T, m(v)m(v)^T\rangle$ where $c(Y)$ are the coefficients of $p(\langle Y, v\rangle)$ and $m(v)$ is the vector containing monomials. The dot product above is the usual trace inner product between matrices. Thus, it is sufficient to show that 
$\Pr\left( \| \E_{D'} c(Y)c(Y)^T - \E_{D} c(Y)c(Y)^T \|_F^2 > \|\E_{D} c(Y)c(Y)^T\|_F^2/4 \right) < 1/d$.
Since $p$ was a univariate polynomial of degree $k$ in $d$ dimensional variables, there are at most $d^{2k}$ entries in total, and each entry is at most a degree $2k$ polynomial of subexponential random variables in $d$ variables. Using standard concentration results for polynomials of subexponential random variables (for instance Theorem 1.2 from \cite{PolyConc} and the references therein). We see that  each entry satisfies 
$\Pr \left( \left| \E_D c(Y)_i c(Y)_j - \E_{D'} c(Y)_i c(Y)_j \right| > \eps \right) \leq \exp\left(- \Omega\left(\frac{n \eps}{\E(c(Y)_i c(Y)_j)^2}\right)^{1/2k}\right)$.
An application of a union bound, squaring the term inside and replacing $\eps^2$ by $\E(c(Y)_i c(Y)_j)^2/4$ gives us 
$\Pr \left( \sum_{i,j = 1}^{d^{2k}} \left( \E_D c(Y)_i c(Y)_j - \E_{D'} c(Y)_i c(Y)_j \right)^2 > \| \E c(Y) c(Y)^T \|^2_F/4  \right) \leq d^{2k} \exp\left(- \Omega\left(\frac{n}{d^{O(k)}}\right)^{1/2k}\right)$.
Hence, setting $n = O((k d \log(d))^{O(k)})$ ensures that with probability at least $1-1/d$, the distribution $D'$ is $(2C, \delta)$-anti-concentrated.

\end{proof}

We say that a $d \times d$ matrix $A$ is $C'$-well-conditioned if all singular values of $A$ are within a factor of $C'$ of each other.

\begin{lemma}[Certified anti-concentration under linear transformations]\label{lem:linear-transformations-anti-concentration}
Let $Y$ be $k$-certifiably $(C,\delta)$-anti-concentrated random variable over $\R^d$.
Let $A$ be any $C'$-well-conditioned linear transformation.
Then, $AY$ is $k$-certifiably $(C,C'^2\delta)$-anti-concentrated.
\end{lemma}
\begin{proof} 
Let $\|A\|$ be the largest singular value of $A$.
Let $p$ be a polynomial that witnesses the certifiable anti-concentration of $Y$.
Let $q(z) = p(z/\|A\|)$.
We will prove that $q$ witnesses the $k$-certifiable $(C,C'^2 \delta)$-anti-concentration of $AY$.

Towards this, observe that:$
\Set{\iprod{Y,v}^2 \leq \delta^2 \E \iprod{Y, v}^2} \sststile{2}{v} \Set{\iprod{AY,v}^2 \leq \delta^2 \E \iprod{AY,v}^2}$.
$\left \{ \langle Y,(A^Tv)/\|A\| \rangle^2 \leq \delta^2 \E \langle Y,(A^Tv)/\|A\| \rangle^2 \right \}  \sststile{k}{v} \left \{ (p(\langle Y,(A^Tv)/\|A\| \rangle) -1)^2\leq \delta^2 \right \}$,

This is the same as $\left \{ \langle AY,v\rangle^2 \leq \delta^2 \E \langle AY,v\rangle^2 \right \}  \sststile{k}{v} \left \{ (q(\langle AY,v\rangle) -1)^2\leq \delta^2 \right \}$, where $q = p(x/\|A\|)$. 
Now, for $w = (A^Tv)/\|A\|$ and any unit vector $v$,
\[ \left \{ \|w\|_2^2 \leq 1 \right \} \sststile{k}{v} \left \{ \|A^Tv\|_2^2 / \|A\|_2^2 \E p^2(\left \langle AY,v\rangle/\|A\| \right) \leq C\delta \right \} \mcom\]
Thus, $\left \{ \|A^Tv\|_2^2 \leq \|A\|^2 \right \} \sststile{k}{v} \left \{ \|A^Tv\|_2^2\E q^2(\left \langle AY,v\rangle \right) \leq C \|A\|_2^2 \delta \right \}$. Using $\left \{ \|v\|_2^2 \leq 1 \right \} \sststile{2}{v} \left \{ \|A^Tv\|_2^2 \leq \|A\|^2 \right \}$,
and thus, 
$\left \{ \|v\|_2^2 \leq 1 \right \} \sststile{k}{v} \left \{ \|v\|_2^2 \E q^2(\left \langle AY,v\rangle \right) \leq C C'^2 \delta \right \}$.
\end{proof} 
\begin{lemma}[Certifiable Anti-Concentration in Boolean Directions]
Fix $C> 0$. Let $Y$ be a $\R^d$ valued \emph{product} random variable satisfying: 
\begin{enumerate}
\item \textbf{Identical Coordinates}: $Y_i$ are identically distributed for every $1 \leq i \leq d$. 
\item \textbf{Anti-Concentration} For every $v \in \left\{0,\pm\frac{1}{\sqrt{d}} \right\}^d$, $\Pr[|\iprod{Y,v}| \leq \delta \sqrt{\E \iprod{Y,v}^2}] \leq C\delta$.
\item \textbf{Light tails} For every $v \in \S^{d-1}$, $\Pr[ |\iprod{Y,v}|> t \sqrt{\E \iprod{Y,v}^2}] \leq \exp(-t^{2}/C)$. 
\end{enumerate}
Then, $Y$ is $k$-certifiably $(C,\delta)$-anti-concentrated for $k = O\left(\frac{\log^{2}(1/\delta)}{\delta^{2}}\right)$.
\end{lemma}
\begin{proof} 
We use the $p$ from Lemma~\ref{lem:univppty_box}. Observe that every monomial of even degree $2k$ for any $k \in \mathbb{N}$, $\E_{Y \sim D} \iprod{Y,v}^{2k}$ is a \emph{symmetric} polynomial in $v$ with non-zero coefficients only on even-degree monomials in $v$. This follows by noting that the coordinates of $D$ are independent and identically distributed and $p$ is an even function. It is a fact that all symmetric polynomials in $v$ can be expressed as polynomials in the ``power-sum'' polynomials $\|v\|_{2i}^{2i}$ for $i \leq 2t$. However, since $v_i^2 \in \left \{0,  \frac{1}{d} \right \}$ for $i \geq 1$, $\|v\|_{2i}^{2i} = \frac{1}{d^{i-1}} \|v\|_2^2 $. Hence a polynomial in $\|v\|_{2i}^{2i}$ is also a univariate polynomial in $\|v\|_2^2$. Since these are polynomial inequalities, they are also sum-of-squares proofs of these inequalities. 

The observation above implies $ \|v\|_2^2 \E_{Y} p(\langle Y, v \rangle)^2 = \|v\|_2^2 \cdot F(\|v\|_2^2)$ for some degree $k$ univariate polynomial $F$. Since  Since $F$ is a univariate polynomial and $\|v\|_2^2\leq 1$ is an ``interval constraint'' by applying Fact~\ref{fact:univariate-interval}, we get:
$\sststile{2t}{\|v\|_2^2} \Set{\|v\|_2^2 F(\|v\|_2^2) \leq C \delta}$.
Recalling the fact that $ \|v\|_2^2 \E_{Y} p(\langle Y, v \rangle)^2 = \|v\|_2^2 \cdot F(\|v\|_2^2)$, this completes the proof. 
\end{proof}






\section{Information-Theoretic Lower Bounds for List-Decodable Regression} \label{sec:lower-bound}
In this section, we show that list-decodable regression on $\Lin_D(\alpha,\ell^*)$ information-theoretically requires that $D$ satisfy $\alpha$-anti-concentration: $\Pr_{x \sim D}[ \iprod{x,v} = 0] < \alpha$ for any non-zero $v$. 

\begin{theorem}[Main Lower Bound] \label{thm:main-lower-bound}
For every $q$, there is a distribution $D$ on $\R^d$ satisfying $\Pr_{x \sim D}[ \iprod{x,v} = 0] \leq \frac{1}{q}$ such that there's no $\frac{1}{2q}$-approximate list-decodable regression algorithm for $\Lin_D(\frac{1}{q},\ell^*)$ that can output a list of size $< d$.  
\end{theorem}
\begin{remark}[Impossibility of Mixed Linear Regression on the Hypercube]
Our construction for the case of $q = 2$ actually shows the impossibility of the well-studied and potentially easier problem of noiseless \emph{mixed linear regression} on the uniform distribution on $\zo^n$. This is because $\cR_i$ is, by construction, obtained by using one of  $e_i$ or $\1-e_i$ to label each example point with equal probability. 
\end{remark}

Theorem~\ref{thm:main-lower-bound} is tight in a precise way. In Proposition~\ref{prop:identifiability}, we proved that whenever $D$ satisfies $\Pr_{x \sim D} [\iprod{x,v} =0] < \frac{1}{q}$, there is an (inefficient) algorithm for \emph{exact} list-decodable regression algorithm for $\Lin_D(\frac{1}{q},\ell^*)$. Note that our lower bound holds even in the setting where there is no additive noise in the inliers. 

Somewhat surprisingly, our lower bound holds for extremely natural and well-studied distributions - uniform distribution on $\zo^n$ and more generally, uniform distribution on $\{0,1,\ldots,q-1\}^d = [q]^d$ for any $q$. We can easily determine a tight bound on the anti-concentration of both these distributions.

\begin{lemma}
For any non-zero $v \in \R^d$, $\Pr_{x \sim \zo^n} \iprod{x,v} = 0 \leq \frac{1}{2}$ and $\Pr_{x \sim [q]^d} [\iprod{x,v} = 0] \leq \frac{1}{q}$.
\end{lemma}
Note that this is tight for any $v = e_i$, the vector with $1$ in the $i$th coordinates and $0$s in all others.
\begin{proof}
Fix any $v$. 
Without loss of generality, assume that all coordinates of $v$ are non-zero. 
If not, we can simply work with the uniform distribution on the sub-hypercube corresponding to the non-zero coordinates of $v$.

Let $S \subseteq \zo^n$ ($[q]^d$, respectively) be the set of all $x \in \zo^n$ ($[q]^d$, respectively) such that $\iprod{x,v} = 0$. 
Then, observe that for any $x \in S$, and any $i$, $x^{(i)}$ obtained by flipping the $i$th bit (changing the $i$th coordinate to any other value) of $x$ cannot be in $S$. 
Thus, $S$ is an independent set in the graph on $\zo^n$ (in $[q]^d$, respectively) with edges between pairs of points with hamming distance $1$. 

It is a standard fact~\cite{wiki-singleton} that the maximum independent set in the $d$-hypercube is of size exactly $2^{d-1}$ and in the $q$-ary Hamming graph $[q]^d$ is of size $q^{d-1}$. 
Thus, $\Pr_{x \sim \zo^d} [\iprod{x,v} = 0] \leq \frac{1}{2}$ and $\Pr_{x \sim [q]^d} [\iprod{x,v} = 0] \leq \frac{1}{q}$.

\end{proof}

To prove our lower bound, we give a family of $d$ distributions on labeled linear equations, $\cR_i$  for  $1 \leq i \leq d$ that satisfy the following: 
\begin{enumerate}
\item The examples in each are chosen from uniform distribution on $[q]^d$, 
\item $\frac{1}{q}$ fraction of the samples are labeled by $e_i$ in $\cR_i$, and, 
\item for any $i, j$, $\cR_i$ and $\cR_j$ are statistically indistinguishable. 
\end{enumerate}
Thus, given samples from $\cR_i$, any $\frac{1}{2q}$-approximate list-decoding algorithm must produce a list of size at least $d$.

Our construction and analysis of $\cR_i$ is simple and exactly the same in both the cases. 
However it is somewhat easier to understand for the  case of the hypercube ($q = 2$).
The following simple observation is the key to our construction.

\begin{lemma}
For $1 \leq i \leq d$, let $\cR_i$ be the distribution on linear equations induced by the following sampling method: Sample $x \sim \zo^d$, choose $a \sim \zo$ uniformly at random and output: $(x, \iprod{x, (1-a)e_i})$. Then, $\cR_i = \cR_j$ for any $i,j \leq  d$. 
\end{lemma}

\begin{proof}
The proof follows by observing that $\cR_i$ when viewed as a distribution on  $\R^{d+1}$ is same as the uniform distribution on  $\zo^{d+1}$ and thus independent of $i$. 
\end{proof}

The argument immediately generalizes to $[q]^d$ and yields:

\begin{lemma}
For $1 \leq i \leq d$, let $\cR_i$ be the distribution on linear equations induced by the following sampling method: Sample $x \sim [q]^d$, choose $a \sim \zo$ uniformly at random and output: $(x, \Paren{\iprod{x, e_i}+a} \text{ mod } q)$. Then, $\cR_i = \cR_j$ for any $i,j \leq  d$. 
\end{lemma}

In this case, we interpret the $1/q$ fraction of the samples where $a = 0$ as the inliers. 
Observe that these are labeled by a single linear  function $e_i$ in any $\cR_i$. 
Thus, they form a valid model in $\Lin_D(\alpha,\ell^*)$ for $\alpha =  1/q$.

Since the linear functions defined by $e_i$ on $[q]^d$, when normalized to have unit norm, have a pairwise Euclidean distance of at least $1/q$, we immediately obtain a proof of  Theorem~\ref{thm:main-lower-bound}.



\section*{Acknowledgement}
We thank Surbhi Goel for pointing out a bug in an earlier version of the paper. 
P.K. thanks David Steurer for illuminating discussions on list-decodable robust estimation via SoS.  


  \phantomsection
  \addcontentsline{toc}{section}{References}
  \bibliographystyle{amsalpha}
  \bibliography{bib/mathreview,bib/dblp,bib/custom,bib/scholar,bib/main}

\appendix

\section{Polynomial Approximation for Core-Indicator}
\label{appendix:polyfact}

The main result of this section is a low-degree polynomial approximator for the function $\1(|x| < \delta)$ with respect to all distributions that have strictly sub-exponential tails.

\begin{lemma}\label{lem:univppty_box} 
Let $D$ be a distribution on $\R$ with mean $0$, variance $\sigma^2 \leq 1$ and satisfying:
\begin{enumerate}
    \item \textbf{Anti-Concentration:} For all $\eta > 0$, $\Pr_{x \sim D} [ |x| < \eta \sigma] \leq C\eta$, and,
    \item \textbf{Tail bound:} $\Pr[ |x| \geq t \sigma] \leq e^{-\frac{t^{2/k}}{C}}$ for $k <2$ and all $t$,
\end{enumerate} 
for some $C>1$. Then, for any $\delta > 0$, there is a $d = O\Paren{\frac{\log^{(4+k)/(2-k)}(1/\delta)}{\delta^{2/(2-k)}}}= \tilde{O}\left(\frac{1}{\delta^{2/(2-k)}}\right)$ and an even polynomial $q(x)$ of degree $d$ such that $q(0) = 1$, $q(x)= 1 \pm \delta$ for all $|x| \leq \delta$ and $\sigma^2 \cdot \E_{x \sim D}\left[ q^2(x) \right] \leq 10C\delta$.
\end{lemma}

Before proceeding to the proof, we note that the bounds on the degree above are tight up to poly logarithmic factors for the gaussian distribution. 

\begin{lemma} \label{lem:tightness-gaussian}
For every polynomial $p$ of degree $d$ such that $p(0)= 1$, $\E_{x \sim \cN(0,1)} [ p^2(x)] = \Omega\left(\frac{1}{\sqrt{d}} \right)$. Further, there is a polynomial $p_*$ of degree $d$ such  that $p_{*}(0) = 1$ and $\E_{x \sim \cN(0,1)} p_{*}^2(x) = \Theta\left(\frac{1}{\sqrt{d}}\right)$.
\end{lemma}

Our construction of the polynomial is based on standard techniques in approximation theory for constructing polynomial approximators for continuous functions over an interval. Most relevant for us are various works of Eremenko and Yuditskii~\cite{MR2441914,MR2837087,MR2346548} and Diakonikolas, Gopalan, Jaiswal, Servedio and Viola~\cite{DGJSV09} on such constructions for the sign function on the interval $[-1,a] \cup [a,1]$ for $a >0$. We point the reader to the excellent survey of this beautiful line of work by Lubinsky~\cite{2007math......1099L}. 

\begin{fact}[Theorem 3.5 in ~\cite{DGJSV09}]
\label{fact:boxpoly}
Let $0 < \eta < 0.1$, then there exist constants $C, c$ such that for 
\[ a:= \eta^2 / C\log(1/\eta)\text{ and }K = 4c\log(1/\eta)/a + 2 < O(\log^2(1/\eta)/\eta^2)\]
there is a polynomial $p(t)$ of degree $K$ satisfying 
\begin{enumerate}
    \item $p(t) > \sign(t) > -p(-t)$ for all $t \in \mathbb{R}$.
    \item $p(t) \in [\sign(t), \sign(t) + \eta]$ for $t \in [-1/2, -2a] \cup [0, 1/2]$.
    \item $p(t) \in [-1, 1+\eta]$ for $t \in (-2a, 0)$
    \item $|p(t)| \leq 2 \cdot (4t)^K$ for all $t > \frac{1}{2}$. 
\end{enumerate}
\end{fact}

We will  also rely on the following elementary integral estimate.
\begin{lemma}[Tail Integral] \label{lem:tail-estimate}
\[\int_{[L, \infty]} \exp\left(-\frac{x^{2/k}}{C}\right) x^{2d}  dx < \exp\left(-\frac{L^{2/k}}{C}\right) ((L)^{4d} + (16kd)^{kd}) \mper \]
\end{lemma}
\begin{proof}
We first prove the claim for $k =1$. 
Let $y = x-L$. The,  $\int_{L}^{\infty} e^{-x^2} x^{2d} dx = \int_{0}^{\infty} e^{-(y+L)^2} (y+L)^{2d} dy$.
We now use that $y^2 + L^2  \leq (y+L)^2$ for all $y \geq 0$ and $(y+L)^{2d} \leq 2^{2d} (y^{2d} + L^{2d})$ to upper bound the integral above by: $e^{-L^2} L^{2d}  + 2^{2d} e^{-L^2} \int_{0}^{\infty} e^{-y^2} y^{2d}$. Using $\int_{0}^{\infty} e^{-y^2} y^{2d} < (4d)^{d}$ gives a bound of $e^{-L^2} (L^{2d}+ (8d)^d)$. 

For larger $k$, we substitute $y = x^{1/k}$ and write the integral in question as $\int_{L^{1/k}}^{\infty} e^{-y^2} y^{2kd-(k-1)} dy$. 
Applying the calculation from the above special case, this integral is upper bounded by: $e^{-L^{2/k}} (L^{4d} + (16kd)^{kd})$.
\end{proof}

\begin{proof}[Proof of Lemma~\ref{lem:univppty_box}]
Let $p(x)$ be the degree $d < O\left( \frac{L\log^2(1/\delta)}{\delta}\right)$ polynomial from Fact~\ref{fact:boxpoly}. We then construct a polynomial $q(x)$ that will be close to $0$ in the range $[\delta, L]$ and $[-L, -\delta]$ and close to $1$ in the range $[-\delta, \delta]$. Our polynomial $q$ is obtained by shifting and appropriately scaling two copies of $p$. 
\[
    q(x) = \frac{p\left(a+\frac{x}{4L}\right) + p\left(-(a+\frac{x}{4L})\right) - 1}{p\left(a\right) + p\left(-a\right) - 1}\\
\]
Then, $q(0) = 1$. It further satisfies:
\begin{enumerate}
    \item $q(x) \in [0, C \sqrt{\delta/L}]$ for $x \in [\delta, L]\cup[-L,\delta]$.
    \item $q(x) \in [1-C  \sqrt{\delta/L}, 1+ \sqrt{\delta/L}]$ for $x \in [-\delta, \delta]$.
    \item $q(x) \in [0, 1+ \sqrt{\delta/L}]$ for $x \in [-3\delta, -\delta]\cup[\delta, 3\delta]$.
    \item $|q(x)| < 4 \cdot (4x)^{t}$ for $|x| > L$
\end{enumerate}
We now prove the bound the $\E p^2$. We do this by providing upper bounds on the contributions to $\sigma^2 \cdot \E_{x \sim \calD}\left[ q^2(\sigma x) \right]$ from the disjoint sets with different guarantees below. Since we are going to evaluate $q(\sigma x)$ the intervals will be scaled by $\sigma$. 
The contributions from the regions $\frac{1}{\sigma}[\delta, L]$ and $\frac{1}{\sigma}[-\delta, \delta]$ can be naively upper bounded by the maximum value that the polynomial can take here times the probability of landing in these regions. The first of these contributes $\sigma \cdot \frac{\delta}{L} \cdot \left(L - \delta\right) \leq \delta$, and using anticoncentration, the second region contributes $ \left(1+\sqrt{\frac{\delta}{L}}\right)^2 \cdot 2C \delta \leq 4C \delta$. The region $\frac{1}{\sigma} [\delta, 3\delta]$ can be bounded similarly to get an upper bound of $2\left(1+\sqrt{\frac{\delta}{L}}\right)^2 \sigma^2 \delta \leq 4 \delta$. To finish, we use Lemma~\ref{lem:tail-estimate} to upper bound the contribution to $\E p^2$ from the tail:
\begin{align*}
    \sigma^2 C' \int_{\frac{1}{\sigma} [L, \infty]} q^2(\sigma x)\exp\left(-\frac{x^{2/k}}{C}\right) dx &\lesssim \sigma^{2+d}  4^d\exp\left(-\frac{1}{C} \cdot \left(\frac{L}{\sigma}\right)^{2/k}\right) ((L/\sigma)^{4d} + (16kd)^{kd})\\
    &\lesssim \exp\left( 2d + 4d\log\left(\frac{L}{\sigma}\right) -\frac{1}{C} \cdot \left(\frac{L}{\sigma}\right)^{2/k} + kd \log(16kd)\right)\mper
\end{align*}
We choose $L$ satisfying $10 d \log(d)+ 4d\log(\frac{L}{\sigma}) -\frac{1}{C} \cdot (\frac{L}{\sigma})^{2/k} < 2 \log(1/\delta)$.

Since $d = O\left( \frac{L\log^2(1/\delta)}{\delta}\right)$, $k < 2$, and $\sigma < 1$ we can now choose $L = \left(\frac{C100\log^{3}(1/\delta)}{\delta}\right)^{k/(2-k)}$ to satisfy the inequality above and to get $d \lesssim \frac{\log^{2 + 3k/(2-k)}(1/\delta)}{\delta^{1 + k/(2-k)}}$.
 When $k =1$ we get $d = \tilde{O}(1/\delta^2)$. Since $\sigma < 1$ in all the above calculations, we get our result by re-scaling $\delta$. 
\end{proof} 
We now complete the proof of Lemma~\ref{lem:tightness-gaussian}.
\begin{proof} [Proof of Lemma~\ref{lem:tightness-gaussian}]
Any polynomial $p$ of degree $d$ can be written as $p(x) = \sum_{i = 1}^d \alpha_i h_i(x)$ where $h_i$ denote the hermite polynomials of degree $i$, satisfying $\E_{x \sim \cN(0,1)} h_i = 0$ and $\E_{x \sim N(0,1)} [h^2_i(x)] = 1$. Since $p(0) =  1$, using Cauchy-Schwartz inequality, we obtain:
\begin{align*}
\E_{x \sim N(0,1)} [p^2(x)] \cdot \sum_{i=1}^d h_i^2(0) &= \left( \sum_{i=1}^d \alpha_i^2\right) \cdot \left(\sum_{i=1}^d h_i^2(0)\right) \geq \left(\sum_{i=1}^d \alpha_i h_i(0) \right)^2 \geq 1 
\end{align*}
Further, observe that for the polynomial $p_{*}(x)=  \frac{1}{\sum_i h_i^2(0)} \sum_{i} h_i(0) h_i(x)$, the above inequality is tight.  
Using that $h_{2i}(0) = \frac{(2i-1)!!}{\sqrt{(2i)!}}$ and $h_i(0) = 0$ if $i$ is odd, (see, for e.g., \cite{HermiteNum}), we have:
\begin{align*}
    \E_{x \sim N(0,1)} [p^2(x)] &\geq \E_{x \sim \cN(0,1)} p_*^2(x)=\left(\sum_{i=1}^d h_i^2(0)\right)^{-1} = \left(\sum_{i=1}^{d/2} \left(\frac{(2i-1)!!}{\sqrt{(2i)!}}\right)^2\right)^{-1}\\
    & = \left(\sum_{i=1}^{d/2} \frac{(2i)!}{2^{2i} i!^2} \right)^{-1} = \left( \sum_{i=1}^{d/2} \binom{2i}{i} \cdot \frac{1}{2^{2i}} \right)^{-1} = \Theta\left(\sum_{i=1}^{d/2} \frac{1}{\sqrt{i}}\right)^{-1} = \Theta\left(\sqrt{d} \right)^{-1}.
\end{align*}\end{proof}

\section{Brute-force search can generate a $\exp(d)$ size list}

In the following, we write $e_i$ to denote the vector with $1$ in the $i$th coordinate and $0$s in all others. 
\begin{proposition}
There exists a distribution $D$ on $\R^d$ and a model $\Lin_D(\alpha,\ell^*)$ such that for every $\alpha <1/2$, with probability at least $1-1/d$ over the draw of a $n$-size sample $\cS$ from $\Lin_D(\alpha,\ell^*)$, there exists a collection $\Sol \subseteq \Set{S \subseteq \cS \mid |S| = \alpha n}$ of size $\exp(d)$ and unit length vectors $\ell_S$ for every $S \in \Sol$ such that $\ell_S$ satisfies all equations in $S$ and for every $S \neq S' \in \Sol$, $\|\ell_S - \ell_{S'}\|_2 \geq 0.1 $. 

\label{prop:brute-force-doesn't-work}
\end{proposition}
\begin{proof}

Let $D$ be the uniform distribution on $e_1, e_2, \ldots, e_d \in \R^d$. Let $\ell^* := \vec{1}/\sqrt{d}$ be the all-ones vector in $\R^d$ scaled by $1/\sqrt{d}$ and let $d$ samples be drawn from the uncorrupted distribution. These give us our inliers, $\cI = \{ (x_i, y_i) \}_{i=1}^{\alpha n}$. For the outliers, choose the following multiset $\cO :=$ $1/\alpha-1$ copies of $\{ (e_i, j) \mid i \in [d], j \in \{\pm 1/\sqrt{d}\}\}$. This is a sample set of size $2d/\alpha$. Any $a \in \{\pm1/\sqrt{d}\}^d$ is a valid candidate for a solution for this data. This is because for any such $a$, $\cI_a := \{ (e_i, a_i) \mid i \in [d] \} \subset S$ satisfies the following
\begin{enumerate}
    \item $\cI_a \subset S$, $|\cI_a| = d = \frac{\alpha}{2} |S|$ and 
    \item for any $(x, y) \in \cI_a$, $y = \langle x, a \rangle$. 
\end{enumerate}
The Gilbert–Varshamov bound from coding theory now tells us that there are at least $\Omega(\exp(\Omega(d)))$ $\{0, 1\}$ vectors in $d$ dimensions that pairwise have a hamming distance of $0.1 \cdot d$. This transfers to the set $\{ \pm 1 /\sqrt{d} \}$ to give us that there are $\Omega(\exp(\Omega(d)))$ vectors in $\{ \pm 1 /\sqrt{d} \}$ that are pairwise $0.1$ apart in $2$-norm. 

\end{proof}

\end{document}